\documentclass{article}

     \usepackage[nonatbib,preprint]{cdml}

\usepackage[utf8]{inputenc} 
\usepackage[T1]{fontenc}    
\usepackage{hyperref}       
\usepackage{url}            
\usepackage{booktabs}       
\usepackage{amsfonts}       
\usepackage{nicefrac}       
\usepackage{microtype}      

\usepackage{amsbsy}
\usepackage{amsmath}
\usepackage{amsthm}

\usepackage{color}
\usepackage[all]{xy}

\usepackage{wrapfig}

\usepackage{epsf}
\usepackage{epsfig}

\usepackage[switch]{lineno}  %

\usepackage{enumitem}

\newtheorem{theorem}{Theorem}

\usepackage[vlined,linesnumberedhidden,ruled,resetcount]{algorithm2e}

\SetKwBlock{Repeat}{repeat}{}
\SetKwInOut{Initialization}{Initialization}

\newcommand{\bfX}{\boldsymbol{X}}

\newcommand{\bfW}{\boldsymbol{W}}
\newcommand{\bfA}{\boldsymbol{A}}
\newcommand{\ci}{\perp\!\!\!\perp}

\newcommand{\bfC}{\boldsymbol{C}}

\newcommand{\bfM}{\boldsymbol{M}}
\newcommand{\bfZ}{\boldsymbol{Z}}
\newcommand{\bfU}{\boldsymbol{U}}
\newcommand{\bfI}{\boldsymbol{I}}
\newcommand{\bfB}{\boldsymbol{B}}

\newcommand{\bfS}{\boldsymbol{S}}
\newcommand{\bfV}{\boldsymbol{V}}

\newcommand{\bfs}{\boldsymbol{s}}
\newcommand{\bfv}{\boldsymbol{v}}

\newcommand{\bfTheta}{\boldsymbol{\Theta}}
\newcommand{\bfGamma}{\boldsymbol{\Gamma}}
\newcommand{\bfomega}{\boldsymbol{\omega}}

\newcommand{\bfBeta}{\boldsymbol{\beta}}
\newcommand{\bfSigma}{\boldsymbol{\Sigma}}

\newcommand{\bfzero}{\boldsymbol{0}}

\newcommand{\beginsupplement}{%
        \setcounter{table}{0}
        \renewcommand{\thetable}{S\arabic{table}}%
        \setcounter{figure}{0}
        \renewcommand{\thefigure}{S\arabic{figure}}%
     }

\title{Towards causality-aware predictions in static anticausal machine learning tasks: the linear structural causal model case}

\author{%
  Elias Chaibub Neto \\
  Sage Bionetworks, Seattle, WA 98121 \\
  \texttt{elias.chaibub.neto@sagebase.org}
}

\begin{document}

\maketitle
\begin{abstract}
We propose a counterfactual approach to train ``causality-aware" predictive models that are able to leverage causal information in static anticausal machine learning tasks (i.e., prediction tasks where the outcome influences the features). In applications plagued by confounding, the approach can be used to generate predictions that are free from the influence of observed confounders. In applications involving observed mediators, the approach can be used to generate predictions that only capture the direct or the indirect causal influences. Mechanistically, we train supervised learners on (counterfactually) simulated features which retain only the associations generated by the causal relations of interest. We focus on linear models, where analytical results connecting covariances, causal effects, and prediction mean squared errors are readily available. Quite importantly, we show that our approach does not require knowledge of the full causal graph. It suffices to know which variables represent potential confounders and/or mediators. We discuss the stability of the method with respect to dataset shifts generated by selection biases and validate the approach using synthetic data experiments.
\end{abstract}

\section{Introduction}

Causal modeling has been recognized as a potential solution to many challenging problems in machine learning (ML)~\cite{pearl2019}. Current approaches operating at the intersection between causality and ML can be roughly split into three different classes. The first, focus on the prediction of the consequences of different actions, policies, and interventions, aiming to improve decision making. These approaches attempt to answer ``what if" counterfactual questions such as ``What if I had treated a patient differently?". The second class focus on the generation of invariant/stable predictions aiming to improve model generalization under dataset shifts~\cite{quionero2009}, while the third class is largely concerned with the estimation of causal effects and only uses ML techniques as a tool to improve the estimation of causal effects. (These approaches will be reviewed in more detail in the Related work section.)

In this paper, our goal is to generate causality-inspired predictions that only leverage associations generated by the causal mechanisms that we are interested in modeling. To this end, we propose a simple counterfactual approach to train ``causality-aware" predictive models, where we train and evaluate ML algorithms on (counterfactually) simulated features which retain only the associations of interest. For instance, in anticausal prediction tasks influenced by mediators and/or confounders where we are interested in the direct effects of the outcome on the features, we simulate counterfactual features containing only the associations generated by the direct causal effects. This ability to generate learners that only leverage associations generated by the causal relations of interest is important in practice. For instance, in situations where confounding is unstable across the training and target populations (while direct causal effects are stable), the approach can be used to generate more stable predictions. Furthermore, in situations where the confounders and/or mediators represent sensitive variables, the approach can also be used to generate predictions that are free from the direct influence of the sensitive variables\footnote{The approach can also be used to generate predictions that are exclusively driven by associations generated by sensitive variables. Such models could be used, for example, to demonstrate how the sensitive variables can still impact the predictive performance of a learner, even when they are not included as inputs in the model.}. (In this paper, however, we present synthetic data illustrations focusing on stable prediction applications, rather than on the analysis of sensitive variables.)

We focus on linear models, where analytical results connecting covariances, causal effects, and prediction mean squared error (MSE) are readily available. At first sight, the proposed approach appears to require the strong assumption that one needs to know the full causal graph describing the data generation process. We point out, however, that this is not the case. The approach only requires partial domain knowledge about which variables represent potential confounders and/or mediators. Noteworthy, we will describe how we can always reparameterize the model in a way that the covariance generated by the causal relations among the features is pushed towards the feature error terms (and similarly for the covariances among the mediators and the covariances among the confounders) so that we can safely generate counterfactual data without even knowing how these variables are causally related. In practice, this is an important advantage in applications involving high-dimensional feature spaces and metadata, where it is unlikely that domain knowledge about these causal relationships will be available.

We also investigate the stability of the proposed approach with respect to (w.r.t.) dataset shifts~\cite{quionero2009}. A standard assumption in supervised ML is that the training and test sets are independent and identically distributed. In practice, however, this assumption is often violated, and dataset shifts are commonly observed in the real world. At the same time, ML models are often capable of leveraging subtle statistical associations between the input ($\bfX$) and outcome ($Y$) variables in the training data, including spurious associations generated by confounders ($\bfC$) and other sources of biases in the data. As a consequence, predictions from confounded learners are often unstable across shifted test sets, and can fail to generalize.

We focus on dataset shifts generated by selection biases~\cite{heckman1979,hernan2004,bareinboim2012} affecting the joint distribution of the confounders and outcome variable, $P(\bfC, Y)$. In real word applications, selection biases often lead to the collection of non-representative training sets and represent an important challenge for ML. While simple approaches such as matching and inverse probability weighting can be used to neutralize these issues in situations where the joint distribution of $\bfC$ and $Y$ in the target population is known, here we focus on the case where the test set can be shifted in unknown ways w.r.t. $P(\bfC, Y)$. This more challenging setting requires more sophisticated adjustment methods, which are sometimes applied to the training data alone with the hope that training an unconfounded model will be enough to generate stable predictions in shifted test sets. Here, we show that this is insufficient, and that deconfounding both the training and test set features can produce more stable predictions.

\vspace{-0.2cm}
\section{Related work}
\vspace{-0.2cm}
Causal approaches based on counterfactual thinking have been used in the context of ML applications to predict the outcomes of different actions, policies, and interventions using non-experimental data~\cite{bottou2013,swaminathan2015,johansson2016,schulam2017}. The goal is to make ``what if" predictions of the consequences of different actions in order to guide decisions. These approaches, however, are only applicable in situations where the ``treatment" variables correspond to features of the ML model, so that prediction goes in the same direction of the causal effect (i.e., the features influence the response variable). Our approach, on the other hand, focus on static anticausal ML tasks where the response influences the features.

Our work is similar in spirit to invariant prediction approaches~\cite{peters2016,ghassami2017,heinze2018,rojascarulla2018,magliacane2018,irm2019} or stable prediction approaches~\cite{kuang2018,subbaswamy2018,subbaswamy2019,kuang2020} in the sense that it can also be used to generate predictions based on the stable properties of the data, without absorbing unstable spurious associations. Invariant prediction approaches, however, rely on multiple training sets to learn invariances while the causality-aware (and stable prediction) approaches only requires a single training set. Some stable prediction approaches require, nonetheless, full knowledge of the causal graph~\cite{subbaswamy2018}, or can only be directly used in causal prediction tasks~\cite{kuang2018,kuang2020}, while the causality-aware method only requires partial knowledge of the causal graph, and is suited to anticausal tasks. (Supplementary Section 1 provides more detailed discussions on these more closely related approaches.)

Supervised ML has also been extensively used to aid the estimation of causal effects, where it can potentially attenuate model mispecification issues~\cite{kreif2019}. In particular, supervised ML has been used to: (i) improve the calculation of propensity scores~\cite{mccaffrey2004,westreich2010,lee2010,wyss2014,pirracchio2015,zhu2015}; (ii) fit regression approaches to estimate outcome models~\cite{hill2011,austin2012,hahn2017}; and (iii) also for the development of double-robust approaches that combine propensity score and outcome regression approaches together~\cite{gruber2010,chernozhukov}. In this paper, however, we take an opposite strategy where instead of using ML to improve causal inference we leverage (partial) causal knowledge to improve the explainability and robustness of ML predictions.

\section{Preliminaries}
Throughout the text we let $\bfX = (X_1, X_2, \ldots, X_{n_X})^T$, $\bfC = (C_1, C_2, \ldots, C_{n_C})^T$, and $\bfM = (M_1, M_2, \ldots, M_{n_M})^T$ represent, respectively, sets of features, confounders, and mediators,
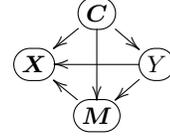
\begin{wrapfigure}{r}{0.15\textwidth}
\vskip -0.2in
$$
{\footnotesize
\xymatrix@-1.4pc{
& *+[F-:<10pt>]{\bfC} \ar[dl] \ar[dr] \ar[dd] & \\
*+[F-:<10pt>]{\bfX} & & *+[F-:<10pt>]{Y} \ar[ll] \ar[dl] \\
& *+[F-:<10pt>]{\bfM} \ar[ul] & \\
}}
$$
\vskip -0.1in
  \caption{Anticausal prediction task.}
  \label{fig:causal.anticausal}
\end{wrapfigure}
while $Y$ represents the response (outcome) variable. The causality-aware counterfactual versions of $\bfX$ and $\bfM$ are represented, respectively, by $\bfX^\ast$ and $\bfM^\ast$. Following~\cite{pearl2009,spirtes2000}, we adopt a mechanism-based approach to causation, where the statistical information encoded in the joint probability distribution of a set of variables is supplemented by a \textit{directed acyclic graph} (DAG) describing our qualitative assumptions about the causal relation between the variables. Following~\cite{scholkopf2012} we denote prediction tasks where the response influences the features as \textit{anticausal prediction tasks}, whereas tasks where the features influence the response are denoted as \textit{causal prediction tasks}. Figure \ref{fig:causal.anticausal} presents the DAG of a general anticausal predictive task, where $\bfX$, $\bfC$, and $\bfM$ are organized into arbitrary DAG subdiagrams (see Supplementary Figure S6 for an illustrative example).

\section{The proposed approach}

\subsection{The univariate case}

For the sake of clarity, we first describe our approach in the special case where $\bfX$, $\bfC$, and $\bfM$ are composed of a single variable. We describe how to use counterfactual reasoning to simulate features where the association between the response and the features is due exclusively to the causal effects of interest. For simplicity, we assume that the data is generated from a standardized linear model\footnote{Note that any linear model $Z^o_k = \mu_k + \Sigma_j \beta_{kj} Z^o_j + U^o_k$, where $Z^o_k$ represents the original data, can be reparameterized into its equivalent standardized form $Z_k = \sum_{j} \theta_{kj} Z_j + U_k$, where $Z_k = (Z^o_k - E(Z^o_k))/\sqrt{Var(Z^o_k)}$ represent standardized variables with $E(Z_k) = 0$ and $Var(Z_k) = 1$; $\theta_{{Z_k}{Z_j}} = \beta_{{Z_k}{Z_j}} \sqrt{Var(Z^o_j)/Var(Z^o_k)}$ represent the path coefficients; and $U_k = U^o_k/\sqrt{Var(Z^o_k)}$ represent the standardized error terms.}, so that the variances of $X$, $C$, $M$, and $Y$ are equal to 1, and the direct causal effect of a variable $Z_j$ on another variable $Z_k$ is represented by the path coefficient~\cite{wright1934}, $\theta_{{Z_k}{Z_j}}$.

The anticausal task presented in Figure \ref{fig:causal.anticausal} is represented by the set of structural equations, $C = U_C$, $Y = \theta_{YC} \, C + U_Y$, $M = \theta_{MC} \, C + \theta_{MY} Y + U_M$, and $X = \theta_{XC} \, C + \theta_{XM} \, M + \theta_{XY} \, Y + U_X$, where $U_C$, $U_Y$, $U_M$, and $U_X$ are independent background (residual) variables.  Using Wright's method of path analysis~\cite{wright1934}, we have that the total covariance (correlation) between $X$ and $Y$,
$$
Cov(X, Y) = \underbrace{\theta_{XY}}_{X \leftarrow Y} + \underbrace{\theta_{XM} \, \theta_{MY}}_{X \leftarrow M \leftarrow Y} + \underbrace{\theta_{XC} \, \theta_{YC}}_{X \leftarrow C \rightarrow Y}~.
$$
can be decomposed into the contribution of the direct causal path, $Y \rightarrow X$, the indirect causal path $Y \rightarrow M \rightarrow X$, and the spurious association generated by the backdoor path $X \leftarrow C \rightarrow Y$. Clearly, the predictive performance of any ML model trained with data generated by this model will be biased by the influence of the confounder $C$ since the learner will leverage the total association between $X$ and $Y$ during training.

Now, suppose that our goal is to build a ML model whose predictive performance is only informed by the direct influence of $Y$ on $X$ and is free from the influence of $C$, as well as, from the indirect influence of $Y$ that is mediated by $M$. To this end, we need to simulate counterfactual data where the association between $X$ and $Y$ is due exclusively to the direct causal effect of $Y$ on $X$. In other words, we want to simulate counterfactual feature data, $X^\ast$, such that $Cov(X^\ast, Y) = \theta_{XY}$. In theory, this could be done by simulating data according to the twin network\footnote{The twin network approach provides a graphical method for evaluating conditional independence relations between counterfactual and factual variables. The basic idea is to use two networks, one representing the factual world and the other the counterfactual world, which share the same background (residual) variables. The factual network (shown to the left of the residual terms) represents the data generation process for the original data, while the counterfactual network (show to the right of the residual terms) shows the modified causal model.}~\cite{balke1994,pearl2009} in Figure \ref{fig:directtwin}, where the new counterfactual feature data, $X^\ast$, is generated from the model $X^\ast = \theta_{XY} Y + U_X$. (In practice, we can estimate $\theta_{XY}$ and $U_X$ by regressing $X$ on $C$, $M$ and $Y$, and simulate the counterfactual feature data using $\hat{X}^\ast = X - \hat{\theta}_{XC} C = \hat{\theta}_{XY} Y + \hat{U}_X$. In other words, we can employ a variation of Pearl's ``abduction, action, prediction" approach to simulate deterministic counterfactuals~\cite{pearl2009,pearl2016}. In the next subsection we explain in detail how the proposed approach differs from Pearl's approach at the ``action" step.) Direct calculation of the covariance between $X^\ast$ and $Y$ shows that,
\begin{align}
Cov&(X^\ast, Y) = Cov(\theta_{XY} Y + U_X, Y) = \theta_{XY} \, Var(Y) + Cov(U_X, Y) = \theta_{XY}~.
\label{eq:anticausal.direct}
\end{align}
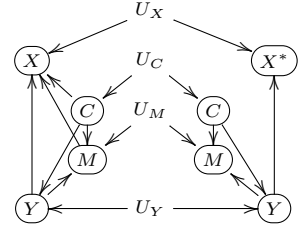
\begin{wrapfigure}{r}{0.27\textwidth}
\vskip -0.1in
{\scriptsize
$$
\xymatrix@-1.4pc{
& & U_X \ar[dll] \ar[drr] & & \\
*+[F-:<10pt>]{X} & & U_C \ar[dl] \ar[dr] & & *+[F-:<10pt>]{X^\ast} \\
& *+[F-:<10pt>]{C} \ar[ul] \ar[ddl] \ar[d] & U_M \ar[dl] \ar[dr] & *+[F-:<10pt>]{C} \ar[rdd] \ar[d] && \\
& *+[F-:<10pt>]{M} \ar[uul] && *+[F-:<10pt>]{M} && \\
*+[F-:<10pt>]{Y} \ar[uuu] \ar[ur] & & U_Y \ar[ll] \ar[rr] & & *+[F-:<10pt>]{Y} \ar[uuu] \ar[ul] \\
}
$$}
\vskip -0.2in
  \caption{Twin network approach in the case where the direct effect represents the causal effect of interest.}
  \label{fig:directtwin}
  \vskip -0.1in
\end{wrapfigure}

Supplementary Section 2 describes the cases where the goal is to build a ML model whose predictive performance is only informed by the indirect causal effect of $Y$ on $X$, as well as, when the goal is to capture the predictive performance informed by the spurious associations generated by the confounder alone. At this point, a natural question is whether alternative interventions would also work. In Supplementary Section 3, we show that a requirement for the intervention to work is that $Y$ is not altered by the intervention. Furthermore, in Supplementary Section 4 we also show that node-splitting transformations in SWIGs~\cite{richardson2013} can also be used as alternative interventions.

\noindent \textbf{Remarks} It is important to highlight that our proposed interventions are different from Pearl's atomic $do(Z=z)$ interventions, and that our counterfactual approach is implemented using a modification of Pearl's ``abduction, action, prediction" procedure for the computation of deterministic counterfactuals. While in Pearl's approach the action step is enforced by a $do(Z=z)$ intervention, where the causal structural model $Z = f(pa(Z), U_Z)$ is replaced by $Z = z$, our interventions are different. For instance, in the case where the direct effect represents the causal effect of interest, our intervention corresponds to replacing $X = f_X(pa(X), U_X) = f_X(C, M, Y, U_X)$ by $X = f_X(pa(X) \setminus \{C \cup M\}, U_X) = f_X(Y, U_X)$. (Note that while our interventions at the action step differs from Pearl's approach, the abduction and prediction steps are still the same.) Also, from a more ``philosophical" point of view, note that even though our proposed interventions represent a different type of microsurgery on the structural causal models, they are still consistent with Lewis' framework of possible worlds~\cite{lewis2013}. Instead of considering counterfactual worlds that develop from different actions than the actions taken in the factual world, our approach considers counterfactual worlds where the data generation mechanisms/laws are different from the mechanisms/laws of the factual world\footnote{As an example, consider an anticausal prediction task described by the DAG $\xymatrix@-1.0pc{C \ar[r] \ar@/^0.5pc/[rr] & X & Y \ar[l]}$, where $Y$ represents the severity score of a disease, $X$ represents a symptom, $C$ represents age, and where the goal is to predict $Y$ using $X$, after removing the spurious association generated by $C$. In our proposed approach, we consider a counterfactual world, $\xymatrix@-1.0pc{C \ar@/^0.5pc/[rr] & X^\ast & Y \ar[l]}$, where age no longer influences the symptom $X$. Note that this intervention can be seen as a type of soft or stochastic intervention where the data generation process differs from the natural system only in the mechanism associated with the feature $X$. Related types of soft/stochastic interventions have been studied in~\cite{correiabareiboim2020,kocaoglu2019,eberhardt2007,malinsky2018}.}. Observe, as well, that our interventions operate at the population level, rather than at the individual level.

\subsection{The multivariate case}

Next, we extend our results to the multivariate case, where the nodes $\bfX$, $\bfC$, and $\bfM$ in Figure \ref{fig:causal.anticausal} represent arbitrary DAG subdiagrams. But first, we describe how we can always reparameterize linear structural causal models in a way that, in practice, we do not need to know how the DAG subdiagrams are organized in order to estimate the causal effects and the residuals employed in the computation of the counterfactual data.

\subsubsection{Reparameterization in linear models}
For linear structural causal models, we can always reparameterize any arbitrary DAG model to a simpler model where the covariance structure between the observed variables is ``pushed" to the unobserved error terms. Figure \ref{fig:dagsimplification} provides an illustrative example of this well-known fact in the structural equations modelling literature~\cite{sobel1987,bollen1989}.
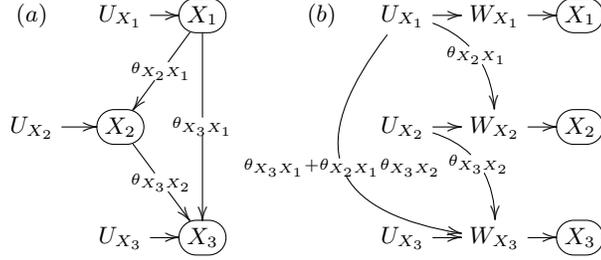
\begin{wrapfigure}{r}{0.58\textwidth}
{\footnotesize
$$
\xymatrix@-1.1pc{
(a) & U_{X_1} \ar[r] & *+[F-:<10pt>]{X_1} \ar[ddl]|-{\theta_{{X_2}{X_1}}} \ar[dddd]|-{\theta_{{X_3}{X_1}}} && (b) & U_{X_1} \ar[r] \ar@/^1.0pc/[ddr]|-{\theta_{{X_2}{X_1}}} \ar@/_3.6pc/[ddddr]|-{\theta_{{X_3}{X_1}} + \theta_{{X_2}{X_1}} \theta_{{X_3}{X_2}}} & W_{X_1} \ar[r] & *+[F-:<10pt>]{X_1} \\
&& \\
U_{X_2} \ar[r] & *+[F-:<10pt>]{X_2} \ar[ddr]|-{\theta_{{X_3}{X_2}}} &&&&  U_{X_2} \ar[r] \ar@/^1.1pc/[ddr]|-{\theta_{{X_3}{X_2}}} & W_{X_2} \ar[r] & *+[F-:<10pt>]{X_2} \\
&& \\
& U_{X_3} \ar[r] & *+[F-:<10pt>]{X_3} &&&  U_{X_3} \ar[r] & W_{X_3} \ar[r] & *+[F-:<10pt>]{X_3} \\
}
$$}
\vskip -0.1in
  \caption{Original (a) and reparameterized (b) models.}
  \label{fig:dagsimplification}
\end{wrapfigure}
The DAG in panel a represents the actual data generation process for the variables $\bfX = (X_1, X_2, X_3)^T$, where the error terms $\bfU_X = (U_{X_1}, U_{X_2}, U_{X_3})^T$ are independent, whereas the DAG in panel b shows the reparameterized model with correlated error terms $\bfW_X = (W_{X_1}, W_{X_2}, W_{X_3})^T$. The set of linear structural causal models describing the DAG in Figure \ref{fig:dagsimplification}a is given by, $\bfX = \bfTheta_{XX} \, \bfX + \bfU_X$, which can be reparameterized as $\bfX = \bfW_X$, where $\bfW_X = (\bfI - \bfTheta_{XX})^{-1} \bfU_X$\footnote{Explicitly, we have that,
{\tiny
\begin{equation*}
\underbrace{
\begin{pmatrix}
X_1 \\
X_2 \\
X_3 \\
\end{pmatrix}}_{\bfX}
=
\underbrace{
\begin{pmatrix}
0 & 0 & 0 \\
\theta_{{X_2}{X_1}} & 0 & 0 \\
\theta_{{X_3}{X_1}} & \theta_{{X_3}{X_2}} & 0 \\
\end{pmatrix}}_{\bfTheta_{XX}}
\underbrace{
\begin{pmatrix}
X_1 \\
X_2 \\
X_3 \\
\end{pmatrix}}_{\bfX}
+
\underbrace{
\begin{pmatrix}
U_{X_1} \\
U_{X_2} \\
U_{X_3} \\
\end{pmatrix}}_{\bfU_X}~, \;\;\;
\underbrace{
\begin{pmatrix}
W_{X_1} \\
W_{X_2} \\
W_{X_3} \\
\end{pmatrix}}_{\bfW_X}
=
\underbrace{
\begin{pmatrix}
1 & 0 & 0 \\
\theta_{{X_2}{X_1}} & 1 & 0 \\
\theta_{{X_3}{X_1}} + \theta_{{X_2}{X_1}} \theta_{{X_3}{X_2}} & \theta_{{X_3}{X_2}} & 1 \\
\end{pmatrix}}_{(\bfI - \bfTheta_{XX})^{-1}}
\underbrace{
\begin{pmatrix}
U_{X_1} \\
U_{X_2} \\
U_{X_3} \\
\end{pmatrix}}_{\bfU_X}~.
\end{equation*}}
Note that because model $\bfX = \bfW_X$ is just a reparameterization of model $\bfX = \bfTheta_{XX} \, \bfX + \bfU_X$, we have that the association structure between the $X_j$ variables is still the same after the model reparameterization. Observe, as well, that for any arbitrary DAG, the matrix $(\bfI - \bfTheta_{XX})$ is always invertible (as fully explained in Supplementary Section 5.1).}.

Next, we describe the above reparameterization for the arbitrary anticausal predictive task. From the DAG in Figure \ref{fig:causal.anticausal}, we have that the joint distribution of the anticausal prediction tasks is factorized as,
\begin{equation*}
P(\bfC, Y, \bfM, \bfX) \, = \, P(\bfC) \, P(Y \mid \bfC) \, P(\bfM \mid \bfC, Y) \, P(\bfX \mid \bfC, \bfM, Y)~,
\end{equation*}
where the components of this factorization are described, respectively, by the structural causal models,
\begin{align*}
\bfC &= \bfTheta_{CC} \, \bfC + \bfU_C~, \hspace{1.0cm} Y = \bfTheta_{YC} \, \bfC + U_Y~, \\
\bfM &= \bfTheta_{MM} \, \bfM + \bfTheta_{MC} \, \bfC + \bfTheta_{MY} \, Y + \bfU_M~, \\
\bfX &= \bfTheta_{XX} \, \bfX + \bfTheta_{XC} \, \bfC + \bfTheta_{XM} \, \bfM + \bfTheta_{XY} \, Y + \bfU_X~,
\end{align*}
where $\bfU_C$, $U_Y$, $\bfU_M$, and $\bfU_X$ are vectors of independent error terms with zero mean and finite variance; $\bfTheta_{CC}$, $\bfTheta_{MM}$, and $\bfTheta_{XX}$ represent, respectively, square matrices of dimension $n_C \times n_C$, $n_M \times n_M$, and $n_X \times n_X$, containing the path coefficients connecting the confounders among themselves, the mediators among themselves and the features among themselves; and $\bfTheta_{YC}$, $\bfTheta_{MC}$, $\bfTheta_{MY}$, $\bfTheta_{XC}$, $\bfTheta_{XM}$, and $\bfTheta_{XY}$, represent retangular matrices of path coefficients connecting variables from separate sets. (For instance, $\bfTheta_{MC}$, corresponds to a $n_M \times n_C$ matrix of path coefficients connecting confounder variables to mediator variables, whereas $\bfTheta_{XY}$, corresponds to a $n_X \times 1$ matrix of path coefficients connecting the response to the features.)

Using simple algebraic manipulations, we can re-write the above linear structural models as,
\begin{align*}
\bfC &= \bfW_C~, \hspace{1.0cm} Y = \bfGamma_{YC} \, \bfC + W_Y~, \\
\bfM &= \bfGamma_{MC} \, \bfC + \bfGamma_{MY} \, Y + \bfW_M~, \\
\bfX &= \bfGamma_{XC} \, \bfC + \bfGamma_{XM} \, \bfM + \bfGamma_{XY} \, Y + \bfW_X~,
\end{align*}
where $W_Y = U_Y$, and $\bfW_V = (\bfI - \bfTheta_{VV})^{-1} \bfU_V$ for $V$ equal to $C$, $M$, or $X$, and $\bfGamma_{YC} = \bfTheta_{YC}$, and $\bfGamma_{ZV} = (\bfI - \bfTheta_{ZZ})^{-1} \bfTheta_{ZV}$ for $\{Z,V\}$ pairs equal to $\{M,C\}$, $\{M,Y\}$, $\{X,C\}$, $\{X,M\}$, and $\{X,Y\}$. Supplementary Section 5 presents a concrete illustrative example of the above reparameterization.

\subsubsection{Estimation of causal effects and residuals in the reparameterized model}

In practice, our counterfactual approach requires the estimation of causal effects and residuals using regression models. For an anticausal task, we regress each feature $X_j$, $j = 1, \ldots, n_X$, on the set of observed confounders and mediators using the regression equations, $X_j = \sum_{k=1}^{n_C} \gamma_{{X_j}{C_k}} C_k + \sum_{k=1}^{n_M} \gamma_{{X_j}{M_k}} M_k + \gamma_{{X_j}{Y}} \, Y + W_{X_j}$, to estimate the causal effects $\hat{\gamma}_{{X_j}{C_k}}$, $\hat{\gamma}_{{X_j}{M_k}}$, $\hat{\gamma}_{{X_j}{Y}}$, and residuals $\hat{W}_{X_j}$ using least squares\footnote{Here, we assume that the number of samples is larger than the number of covariates in the regression fits, and that multicolinearity is not an issue too. Note that we do not need to assume Gaussian error terms.}, and then generate counterfactual features by adding back the estimated residuals to a linear predictor containing only the causal effects of interest. That is, in order to estimate the predictive performance that is separately due to direct causal effects, indirect causal effects, or confounding, we generate counterfactual features using, respectively, $\hat{\bfX}^\ast = \hat{\bfGamma}_{XY} Y + \hat{\bfW}_X$, $\hat{\bfX}^\ast = \hat{\bfGamma}_{XM} \, \hat{\bfM}^\ast + \hat{\bfW}_X$\footnote{Where, $\hat{\bfM}^\ast = M - \hat{\bfGamma}_{MC} \, C = \hat{\bfGamma}_{MY} \, Y + \hat{\bfW}_M$ is calculated by first fitting the regressing models $M_j = \sum_{k=1}^{n_C} \gamma_{{M_j}{C_k}} C_k + \gamma_{{M_j}{Y}} \, Y + W_{M_j}$, to estimate the causal effects $\hat{\gamma}_{{M_j}{C_k}}$, $\hat{\gamma}_{{M_j}{Y}}$ and error terms $\hat{W}_{M_j}$.}, or $\hat{\bfX}^\ast = \hat{\bfGamma}_{XC} \, \bfC + \hat{\bfW}_X$. Importantly, note that when we regress $X_j$ on $\bfC$, $\bfM$, and $Y$ only the coefficients associated with the parents of $X_j$ in the reparameterized model will be statistically different from zero (for large enough sample sizes). Therefore, in practice, we don't need to know before hand which variables are the parents of $X_j$ in the reparameterized model. The parent set will be learned automatically from the data by the regression model fit. (This, of course, assumes the absence of unmeasured confounders. Supplementary Section 6 provides further remarks on potential identification issues.)

\subsubsection{The connection between covariances and causal effects in the multivariate general case}

Here, we extend the univariate results of Section 4.1 to the multivariate case (see Supplementary Section 7 for the proofs).
\begin{theorem}
Consider an anticausal prediction task:
\begin{enumerate}[leftmargin=*]
\item For causal effects generated by the paths in $Y \rightarrow \bfX$, if $\bfX^\ast$ is given by $\bfX^\ast = \bfGamma_{XY} \, Y + \bfW_X$, then $Cov(\bfX^\ast, Y) = \bfGamma_{XY}$.
\item For causal effects generated by the paths in $Y \rightarrow \bfM \rightarrow \bfX$, if  $\bfX^\ast$ is given by $\bfX^\ast = \bfGamma_{XM} \, \bfM^\ast + \bfW_X$, and $\bfM^\ast = \bfGamma_{MY} \, Y + \bfW_M$, then $Cov(\bfX^\ast, Y) = \bfGamma_{XM} \, \bfGamma_{MY}$.
\item For the spurious associations generated by the paths in $\bfX \leftarrow \bfC \rightarrow Y$, if  $\bfX^\ast$ is given by $\bfX^\ast = \bfGamma_{XC} \, \bfC + \bfW_X$, then $Cov(\bfX^\ast, Y) = \bfGamma_{XC} \, Cov(\bfC) \, \bfGamma_{YC}^T$.
\end{enumerate}
\end{theorem}

The above result, together with the estimation approach described In Section 4.2.2, show that by generating causality-aware counterfactual features, $\bfX^\ast$, and then training and evaluating ML learners on this counterfactual data, we are able to leverage only the associations generated by the causal mechanisms of interest. Quite importantly, because the counterfactual data is estimated from the reparameterized model, the approach does not require full knowledge of the causal graph. It suffices to know which variables are confounders and which are mediators.

\section{Confounding adjustment in anticausal tasks}

\subsection{An algorithmic description for confounding adjustment}

When the goal is confounding adjustment, the causality-aware features are generated according to Algorithm 1.
\begin{algorithm}[!h]
\caption{Causality-aware feature computation in anticausal prediction tasks}\label{alg:counterfactual.adjustment}
\KwData{Training data, $\{\bfX_{tr}, \bfC_{tr}, Y_{tr}\}$; test set features and confounders, $\{\bfX_{ts},\bfC_{ts}\}$.}
\ShowLn \For{each feature $X_j$} {
\ShowLn $\bullet$ Using the training set, estimate regression coefficients and residuals from, $X_{j,tr} = \mu_j^{tr} + \beta_{{X_j} Y}^{tr} \, Y_{tr} + \sum_i \beta_{{X_j} {C_i}}^{tr} \, C_{i,tr} + W_{X_j}^{tr}$, and then compute the respective counterfactual feature as, $\hat{X}_{j,tr}^{\ast} = \hat{\mu}_j^{tr} + \hat{\beta}_{{X_j} Y}^{tr} \, Y_{tr} + \hat{W}_{X_j}^{tr}$. \\
\ShowLn $\bullet$ Using the test set, compute the counterfactual feature, $\hat{X}_{j,ts}^{\ast} = X_{j,ts} - \sum_i \hat{\beta}_{{X_j} {C_i}}^{tr} \, C_{i,ts}$.
}
\KwResult{Counterfactual features, $\hat{\bfX}^{\ast}_{tr}$ and $\hat{\bfX}^{\ast}_{ts}$.}
\end{algorithm}
Observe that the algorithm requires test set confounding data (but not the test set labels). Note that for large sample sizes, and under the assumption that the causal effects are stable between the training and test sets, we have that $\hat{\beta}_{{X_j} {C_i}}^{tr} \approx \hat{\beta}_{{X_j} {C_i}}^{ts}$ so that we can estimate the test set counterfactual features without using test set labels since,
\begin{align*}
X_{j,ts}^{\ast} &= X_{j,ts} - \sum \hat{\beta}_{{X_j} {C_i}}^{tr} \, C_{i,ts} \approx X_{j,ts} - \sum \hat{\beta}_{{X_j} {C_i}}^{ts} \, C_{i,ts} = \hat{\mu}_j^{ts} + \hat{\beta}_{{X_j} Y}^{ts} \, Y_{ts} + \hat{W}_{X_j}^{ts}~.
\end{align*}

\subsection{Dataset shifts generated by selection biases}

In anticausal prediction tasks, dataset shifts in the joint distribution of the confounders and outcome variable, $P(\bfC, Y)$, are often caused by selection biases. The confounded anticausal prediction task influenced by selection bias is described by the causal graph in Figure \ref{fig:anticausal.task}, where the auxiliary
\begin{wrapfigure}{r}{0.19\textwidth}
\vskip -0.1in
$$
\xymatrix@-1.2pc{
 & *+[F-:<10pt>]{\bfC} \ar[dl] \ar[dr] \ar[r] & *+[F]{S} \\
*+[F-:<10pt>]{\bfX}  & & *+[F-:<10pt>]{Y} \ar[ll] \ar[u]}
$$
\vskip -0.1in
\caption{}
\vskip -0.1in
\label{fig:anticausal.task}
\end{wrapfigure}
variable $S$ indicates the presence of a selection mechanism contributing to the association between $\bfC$ and $Y$. (Here, $S$ represents a binary variable which indicates whether the sample was included or not in the dataset, and the square frame around $S$ indicates that our dataset is generated conditional on $S$ being set to 1. Note that, the application of the d-separation criterion~\cite{pearl2009} to the causal graph shows that because $S$ is a collider, we have that, conditional on $S = 1$, the additional path $\bfC \rightarrow S \leftarrow Y$ is open and, therefore, contributes to the association between $\bfC$ and $Y$.) In the stability analysis that we present in the next subsection, we assume that the causal effects $\bfBeta_{XY}$ and $\bfBeta_{XC}$ and the residual covariance, $Cov(\bfU_X)$, are the same across the training and test sets, so that $P(\bfX \mid \bfC, Y)$ is stable. We also assume that the causal effect $\beta_{YC}$ is stable, and that the dataset shifts in $P(\bfC, Y)$ are generated by selection biases.

\subsection{Stability under dataset shifts of $P(\bfC, Y)$ generated by selection biases}

While it might seen intuitive that training a learner on unconfounded data will prevent it from learning the confounding signal and, therefore, will lead to more stable predictions in shifted target populations\footnote{Examples of approaches that only adjust the training data include pre-processing techniques to reduce discrimination in ML~\cite{calders2009,kamiran2012}.}, here we show that adjusting the training data alone is insufficient, and that better stability can be achieved by deconfounding the test set features as well.

Next, we present an analysis of this issue using a toy linear model example (the result, nonetheless, holds for more general linear models, as described in Supplementary Section 8). Consider the causal graph in $\xymatrix@-1.0pc{C \ar[r] \ar@/^0.5pc/[rr] & X & Y \ar[l]}$ where $C = U_C$, $Y = \beta_{YC} \, C + U_Y$, and $X = \beta_{XY} \, Y + \beta_{XC} \, C + U_X$, with $E[U_V] = 0$, $Var(U_V) = \sigma^2_V$, for $V = \{C, Y, X\}$.
The goal is to predict the outcome $Y$ using the feature $X$. Assume without loss of generality that the data has been centered. Let $\hat{Y} = X_{ts} \hat{\beta}_{tr}$ represent the test set prediction from a linear regression model, where $\hat{\beta}_{tr}$ represents the coefficient estimated with the training data, and $X_{ts}$ represents the test set feature. By definition the expected MSE is given by,
\begin{align}
E[(Y_{ts} - \hat{Y})^2] &= E[Y^2_{ts}] + E[\hat{Y}^2] - 2 E[\hat{Y} Y_{ts}] = Var[Y_{ts}] + E[\hat{Y}^2] - 2 Cov(\hat{Y}, Y_{ts}) \nonumber \\
&= Var[Y_{ts}] + \hat{\beta}_{tr}^2 Var[X_{ts}] - 2 \hat{\beta}_{tr} Cov(X_{ts}, Y_{ts})~, \label{eq:expected.mse}
\end{align}
where the expectation is w.r.t. the test set (so that $\hat{\beta}_{tr}$ is a fixed constant w.r.t. the expectation).

For any approach which does not process the test set features we have that,
\begin{align*}
Var(X_{ts}) &= Var(\beta_{XY} Y_{ts} + \beta_{XC} C + U_X) \\
&= \sigma^2_X + \beta_{XY}^2 Var(Y_{ts}) + \beta_{XC}^2 Var(C_{ts}) + 2 \beta_{XY} \beta_{XC} Cov(Y_{ts}, C_{ts})~,
\end{align*}
\begin{align*}
Cov(X_{ts}, Y_{ts}) &= Cov(\beta_{XY} Y_{ts} + \beta_{XC} C_{ts} + U_X, Y_{ts}) = \beta_{XY} Var(Y_{ts}) + \beta_{XC} Cov(Y_{ts}, C_{ts})
\end{align*}
showing that both $Var(X_{ts})$ and $Cov(X_{ts}, Y_{ts})$ depend on $Cov(Y_{ts}, C_{ts})$ (so that the $E[MSE]$ will be unstable under dataset shifts of the association between the confounder and the outcome variable). On the other hand, we have that for the causality-aware approach,
\begin{equation*}
Var(X_{ts}^\ast) = Var(\beta_{XY} Y_{ts} + U_X) = \beta_{XY}^2 Var(Y_{ts}) + \sigma^2_X~,
\end{equation*}
\begin{equation*}
Cov(X_{ts}^\ast, Y_{ts}) = Cov(\beta_{XY} Y_{ts} + U_X, Y_{ts}) = \beta_{XY} Var(Y_{ts})~,
\end{equation*}
do not depend on $Cov(Y_{ts}, C_{ts})$, so that the $E[MSE]$ will be stable w.r.t. this particular type of dataset shift (although, as shown by eq. (\ref{eq:expected.mse}) it will be still influenced by dataset shifts on $Var(Y_{ts})$). Note that this is true even when we apply a confounding adjustment to the training set (a situation where the $\hat{\beta}_{tr}$ estimate is not influenced by the spurious associations generated by the confounder). This explains why it is not enough to deconfound the training features alone. While training a regression model using deconfounded features allows us to estimate deconfounded model weights,  $\hat{\bfBeta}_{tr}$, the prediction $\hat{Y} = \bfX_{ts} \hat{\bfBeta}_{tr}$ is a function of both the trained model $\hat{\bfBeta}_{tr}$ and the test set feature, $\bfX_{ts}$. As a consequence, if we do not deconfound the test set features, the expected MSE will still be influenced by the confounders (since, in anticausal prediction tasks, the original test set features, $X_{j,ts} = \beta_{{X_j}Y} Y_{ts} + \beta_{XC} C_{ts} + U_{X_j}$ are still functions of the confounder variable). This point is described in more general terms in Supplementary Section 9, where we show that the expected value of an arbitrary performance metric is still a function of $C_{ts}$ when the features $X_{j,ts}$ are generated by arbitrary structural causal model $X_{j,ts} = f(Y_{ts}, C_{ts}, U_X^{ts})$, even when we train the ML model using deconfounded training set features, $X_{j,tr}^\ast = f^\ast(Y_{tr}, U_X^{tr})$.

\subsection{Synthetic data experiments}

We illustrate the above points in synthetic data experiments investigating the influence of dataset shifts in $P(\bfC, Y)$ on the predictive performance (measured by MSE). In order to investigate the influence of shifts in $Var(Y_{ts})$ on the prediction stability, we performed two experiments, where $Var(Y_{ts})$ was kept constant in the first, but was allowed to vary in the second. In both experiments, we compared the causality-aware adjustment against two alternative approaches denoted as \textit{baseline 1} and \textit{baseline 2} adjustments. The \textit{baseline 1} adjustment represents approaches that remove the causal effect of the confounders on the features in the training set alone, while \textit{baseline 2} represents approaches that remove the association between the confounders and the output in the training set alone (see Supplementary Section 10 for further details). For completeness we also report results based on the ``\textit{no adjustment}" approach, where no adjustments are applied to the training or test sets.

Each experiment was based on 1,000 replications where, for each replication, we generated training sets with $Var(Y_{ts}) = 1$, $Var(C_{ts}) = 1$, and $Cov(C_{ts}, Y_{ts}) = 0.8$, and 9 distinct test sets showing increasing amounts of dataset shifts in the $P(\bfC, Y)$ relative to the training data. In the first experiment (the fixed $Var(Y_{ts})$ case), this was accomplished by varying $Cov(Y_{{ts}}, C_{{ts}})$ according to $\{0.8$, $0.6$, $0.4$, $0.2$, $0.0$, $-0.2$, $-0.4$, $-0.6$, $-0.8\}$ across the 9 test sets, and by varying $Var(C_{{ts}})$ according to $\{1.00$, $1.25$, $1.50$, $1.75$, $2.00$, $2.25$, $2.50$, $2.75$, $3.00\}$, while keeping $Var(Y_{{ts}})$ fixed at 1. In the second experiment (the varying $Var(Y_{ts})$ case), we varied $Cov(Y_{{ts}}, C_{{ts}})$ as before, but kept $Var(C_{{ts}})$ fixed at 1, while increasing $Var(Y_{{ts}})$ according to $\{1.00$, $1.25$, $1.50$, $1.75$, $2.00$, $2.25$, $2.50$, $2.75$, $3.00\}$ across the test sets.

Our experiments were based on linear models containing 10 features and 1 confounder, and on training and test sets containing 1,000 samples. (See Supplementary Section 10 for further details on the synthetic data generation and simulation parameter choices.) The causal effects $\bfBeta_{XY}$, $\bfBeta_{XC}$, and $\beta_{YC}$ and $Cov(\bfU_X)$ were kept constant across the training and test sets in order to guarantee that $P(\bfX \mid \bfC, Y)$ was stable. 


Figures \ref{fig:mpower.aucs.stability.regr.e1} and \ref{fig:mpower.aucs.stability.regr.e2} report the results for the fixed and varying $Var(Y_{ts})$ cases, respectively. In both figures, panels a to d report boxplots of the MSE scores (y-axis) across 1,000 simulation replications for the 9 test sets (x-axis), while panel e presents a comparison of the stability-errors, defined as the standard deviation of the MSE scores across the 9 test sets in each simulation replication.

Figure \ref{fig:mpower.aucs.stability.regr.e1} reports the results for the first experiment. Note that because we kept $Var(Y_{ts})$ constant across the test sets we see perfect stability for the causality-aware approach (panel a). (Observe that varying $Cov(Y_{ts}, C_{ts})$ and $Var(C_{ts})$ has no influence on the stability of the results, since the expected MSE for the causality-aware approach only depends on $Var(Y_{ts})$.)

Figure \ref{fig:mpower.aucs.stability.regr.e2} reports results for the second experiment based on increasing $Var(Y_{ts})$ values. As expected, we now observe instability in the causality-aware approach too. The causality-aware predictions, however, are still more stable than the predictions from the other approaches.

\begin{figure*}[!h]
\centerline{\includegraphics[width=\linewidth]{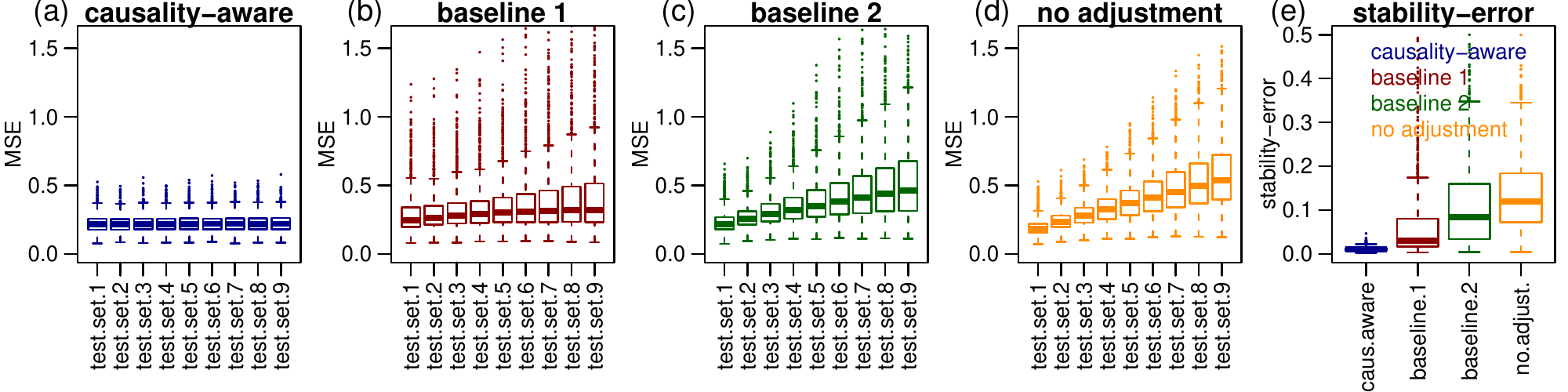}}
\vskip -0.1in
\caption{Synthetic data experiment results for the fixed $Var(Y_{ts})$ case.}
\label{fig:mpower.aucs.stability.regr.e1}
\end{figure*}

\begin{figure*}[!h]
\centerline{\includegraphics[width=\linewidth]{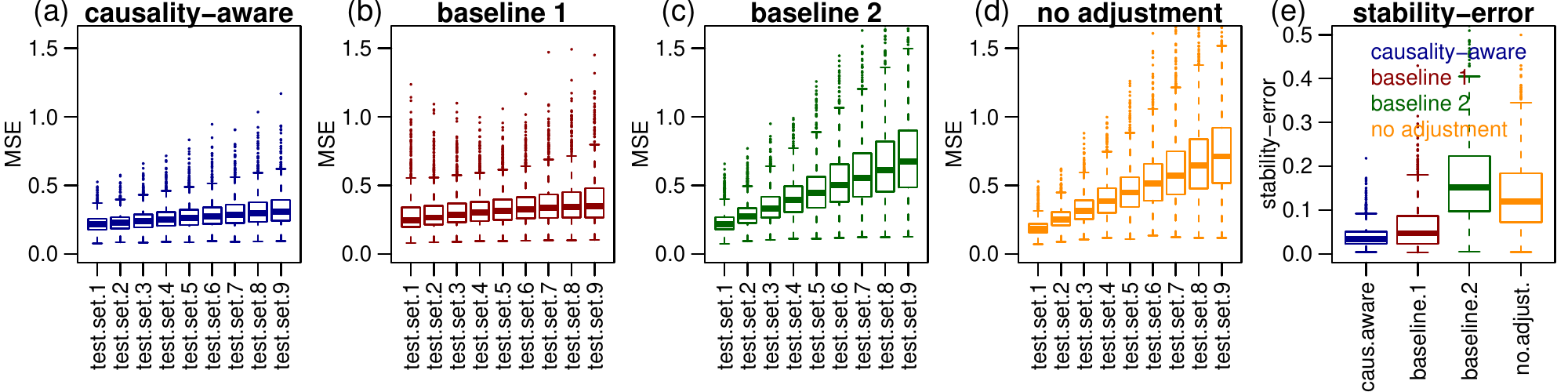}}
\vskip -0.1in
\caption{Synthetic data experiment results for the increasing $Var(Y_{ts})$ case.}
\label{fig:mpower.aucs.stability.regr.e2}
\end{figure*}

\section{Final remarks}

This paper has three main contributions. First, we describe a novel counterfactual approach to train ``causality-aware" predictive models, which leverages only the associations generated by the causal mechanisms of interest. Second, by leveraging a reparameterization of the linear structural causal models (described in Section 4.2.1), we show that the approach does not require full knowledge of the data generation process. It suffices to know which variables are confounders and mediators, without knowing how these variables are causally related. This represent an important practical advantage of the method relative to alternative approaches such as counterfactual normalization~\cite{subbaswamy2018}, which requires knowledge of the full causal graph. Third, we investigate the stability properties of the method w.r.t. dataset shifts generated by selection biases. We show that the $E[MSE]$ for adjustment approaches that fail to deconfound the test set features will be unstable w.r.t. shifts in $Cov(\bfC, Y)$, even when the ML models are trained with unconfounded data (and there are no shifts in $Var(Y_{ts})$). This is an important observation that (we feel) is not well appreciated in the ML community.

One important drawback of the approach is its reliance on the linearity assumption. The present work, however, represents a first step that, we believe, will serve as inspiration for more flexible approaches. Along these lines, in a separate contribution~\cite{achaibubneto2020c} (where we compare the causality-aware approach against the residualization confounding adjustment - an ad-hoc approach, widely used in applied fields such as neuroimaging), we describe an extension of the causality-aware approach to additive models. Furthermore, in another separate contribution~\cite{achaibubneto2020b}, we also describe how the causality-aware approach (based on linear models) can still be used to deconfound the feature representations learned by deep neural network models in classification tasks. The key idea is that by training a highly accurate DNN using softmax activation at the classification layer, we have that, by construction, the feature representation learned by the last layer prior to the output layer will fit well a logistic regression model (since the softmax activation used to classify the outputs of the DNN is essentially performing logistic regression classification). This reference illustrates the practicality of the causality-aware approach in real world applications. (Finally, while this work has focused on anticausal tasks, we present some analogous results for causal prediction tasks in Supplementary Section 11.)

\clearpage

\beginsupplement



\noindent {\Huge Supplement}

\setcounter{section}{0}
\setcounter{theorem}{0}

\section{Further related work}

As clearly articulated by~\cite{subbaswamy2020} there are, broadly speaking, two types of stable prediction approaches: (i) \textit{reactive} methods, that use data (or knowledge) from the intended deployment/target population to correct for shifts; and (ii) \textit{proactive} methods, that do not require data from the deployment/target populations, and are able to learn models that are stable with respect to unknown dataset shifts.

Many reactive approaches in the literature~\cite{shimodaira2000,sugiyama2007,dudik2006,huang2007,gretton2009,bickel2009,liu2014} deal with dataset shift by reweighting the training data to make it more closely aligned it with the target test distribution. In this paper, however, we focus on anticausal prediction tasks~\cite{scholkopf2012} and address only dataset shifts in the joint distribution of the confounders and outcome variable, $P(\bfC, Y)$, caused by selection biases~\cite{heckman1979,hernan2004,bareinboim2012}. In our particular context, we can still use simple reactive approaches when the target (test set) joint distribution, $P(\bfC_{ts}, Y_{ts})$ is known. For instance, if we know, a priori, the prevalence of a disease with respect to a given demographic risk factor in the target population, then we can either subsample or oversample the training data in order to make the training set distribution $P(\bfC_{tr}, Y_{tr})$ match the test set set distribution $P(\bfC_{ts}, Y_{ts})$. In classification tasks, simple balancing approaches, such as matching or approximate inverse probability weighting, can be used to subsample or oversample the training data. In regression tasks, approaches such as propensity scores for continuous variables~\cite{hirano2004}, covariate balancing propensity score methods for continuous variables~\cite{fong2018}, or standard propensity score matching applied to dichotomized outcome data, can be used.


The more challenging case where we face unknown shifts in $P(\bfC_{ts}, Y_{ts})$ (the case we address in this paper) requires more sophisticated adjustment approaches. Several proactive methods have been proposed in the literature. For instance, invariant learning approaches~\cite{peters2016,rojascarulla2018,magliacane2018,irm2019} employ multiple training datasets in order to learn invariant predictions. The causality-aware approach (adopted in this paper), on the other hand, only requires a single training set.

Another proactive approach, which can be applied to anticausal tasks based on a single training set, is the counterfactual normalization method proposed by~\cite{subbaswamy2018}. The approach requires full knowledge of the causal graph describing the data generation process and is implemented in several steps. First, it identifies a set vulnerable variables that make the ML model susceptible to learning unstable relationships that might lead to poor generalization across shifted dataset. Second, the approach performs a node-splitting operation in order to augment the causal graph with counterfactual variables which isolate unstable paths of statistical associations and allow the retention of some stable paths involving vulnerable variables. Third, the approach determines a stable set of input variables that can be used to train a more stable ML model. In practice, the approach is implemented with linear (or additive) models.

Similarly to counterfactual normalization, the causality-aware approach also leverages counterfactual features to improve stability and is also implemented with linear models\footnote{In reference~\cite{achaibubneto2020c} we describe how to causality-aware approach can be extended to additive models.}. There are, nonetheless, important differences. The key idea (in the context of anticausal prediction tasks) is to train and evaluate supervised ML algorithms on counterfactually simulated data which retains only the associations generated by the causal influences of the output variable on the inputs. Noteworthy, as described in the main text, it is always possible to reparameterize the model in a way that the covariance among the features and among the confounders is pushed towards the respective error terms. This allows the generation of counterfactual features without even knowing the causal relations among features and the causal relations among the confounders. As a consequence, the causality-aware approach does not require knowledge of the full data generation process (at least for linear models). Contrary to counterfactual normalization, where the full causal diagram needs to be specified, the causality-aware approach only requires knowledge of which variables are confounders.

Finally, the methods proposed by~\cite{kuang2018,kuang2020} represent another set of stable prediction approaches. The key idea behind these methods is to find a set of covariates for which the expected value of the outcome is stable across distinct test set environments. These covariates fall into two classes: stable variables ($\bfS$) that have an structural relationship with the outcome, and unstable variables ($\bfV$) that can be associated with both the outcome and the stable variables but do not have a causal relation with the outcome. Assuming that there exists a stable function $f(\bfs)$ such that for all testing environments $E(Y \mid \bfS = \bfs, \bfV = \bfv) = E(Y \mid \bfS = \bfs) = f(\bfs)$ - a condition which is fulfilled when $Y \ci \bfV \mid \bfS$ - the approach is able to learn the stable function $f(\bfs)$ without prior knowledge about which variables are stable or unstable. These methods, however, are tailored to causal prediction tasks (i.e., where the inputs have a causal effect on the outcome), and cannot be directly applied in anticausal tasks\footnote{Note that in anticausal prediction tasks $\bfS$ might be a collider. Hence, if $\bfS$ is a collider, it follows that conditional on $\bfS$, $Y$ cannot be independent of $\bfV$, and the assumption $Y \ci \bfV \mid \bfS$ cannot hold.}.

\section{Additional univariate examples}

Consider an anticausal prediction task, and suppose that our goal is to build a ML model whose predictive performance is only informed by the indirect causal effect of $Y$ on $X$. In this case, we simulate data according to the twin network in Figure \ref{fig:indirect.confounding.twin}a, so that,
\begin{align}
Cov(X^\ast, Y) &= Cov(\theta_{XM} M^\ast + U_X, Y) = \theta_{XM} \, Cov(M^\ast, Y) = \theta_{XM} \, Cov(\theta_{MY} Y + U_M, Y) \nonumber \\
&= \theta_{XM} \, \theta_{MY} \, Cov(Y, Y) = \theta_{XM} \, \theta_{MY}~.
\label{eq:anticausal.indirect}
\end{align}

Now, suppose that the goal is to build a ML model whose predictive performance is only informed by the spurious associations generated by the confounder, we can simulate data according to the twin network in Figure \ref{fig:indirect.confounding.twin}b, so that,
\begin{align}
Cov(X^\ast, Y) &= Cov(\theta_{XC} \, C + U_X, \theta_{YC} \, C + U_Y) = \theta_{XC} \, \theta_{YC} \, Cov(C, C) = \theta_{XC} \, \theta_{YC}~.
\label{eq:anticausal.confounding}
\end{align}

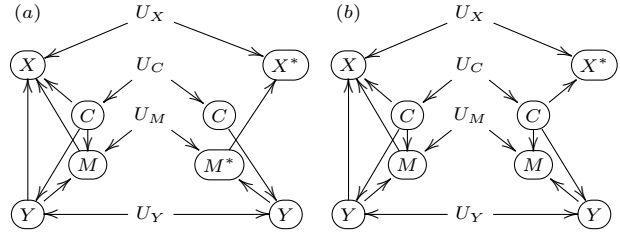
\begin{figure}[!h]
{\scriptsize
$$
\xymatrix@-1.4pc{
(a) & & U_X \ar[dll] \ar[drr] & &  & (b) & & U_X \ar[dll] \ar[drr] & & \\
*+[F-:<10pt>]{X} & & U_C \ar[dl] \ar[dr] & & *+[F-:<10pt>]{X^\ast}  & *+[F-:<10pt>]{X} & & U_C \ar[dl] \ar[dr] & & *+[F-:<10pt>]{X^\ast} \\
& *+[F-:<10pt>]{C} \ar[ul] \ar[ddl] \ar[d] & U_M \ar[dl] \ar[dr] & *+[F-:<10pt>]{C} \ar[rdd] && & *+[F-:<10pt>]{C} \ar[ul] \ar[ddl] \ar[d] & U_M \ar[dl] \ar[dr] & *+[F-:<10pt>]{C} \ar[rdd] \ar[d] \ar[ur] &\\
& *+[F-:<10pt>]{M} \ar[uul] && *+[F-:<10pt>]{M^\ast} \ar[uur] && & *+[F-:<10pt>]{M} \ar[uul] && *+[F-:<10pt>]{M} && \\
*+[F-:<10pt>]{Y} \ar[uuu] \ar[ur] & & U_Y \ar[ll] \ar[rr] & & *+[F-:<10pt>]{Y} \ar[ul]  & *+[F-:<10pt>]{Y} \ar[uuu] \ar[ur] & & U_Y \ar[ll] \ar[rr] & & *+[F-:<10pt>]{Y} \ar[ul] \\
}
$$}
\vskip -0.1in
  \caption{Twin network approach in the case where the indirect effect represents the causal effect of interest (panel a), and in the case where we are interested in estimating predictive performance that is due to confounding effects.}
  \label{fig:indirect.confounding.twin}
\end{figure}

\section{On alternative model modifications for simulating counterfactual data}

In the main text (as well as, in the above section) we showed how to generate counterfactual data that contains only associations generated by the causal effects of interest. A natural question is whether alternative modifications of the causal diagram (other than the ones presented in the main text and in Supplementary Section 2) would also lead to counterfactual datasets containing only the associations due to the causal effects of interest. Here, we show that this is sometimes possible, and clarify that, for anticausal prediction tasks, the requirement for the intervention to work is that $Y$ is not altered by the intervention.

We start with the case where the interest focus on the direct causal effects in anticausal predictive tasks. Here, the goal is to simulate counterfactual data where $Cov(X^\ast, Y^\ast) = \theta_{XY}$. Starting with examples involving confounding alone, consider first an alternative modification where we simulate data with the confounder variable $C$ set to a fixed value $c$, as described in the twin network in Figure \ref{fig:confoundingtwin2}a. Direct calculation shows that,
\begin{align*}
Cov(X^\ast, Y^\ast) &= Cov(\theta_{XC} \, c + \theta_{XY} \, Y^\ast + U_X, Y^\ast) \\
&= \theta_{XY} \, Cov(Y^\ast, Y^\ast) = \theta_{XY} \, Var(Y^\ast) \\
&= \theta_{XY} \, Var(\theta_{YC} \, c + U_Y) = \theta_{XY} \, Var(U_Y) \\
&= \theta_{XY} (1 - \theta_{YC}^2)~,
\end{align*}
for any chosen $c$ value. (Note that $Var(U_Y) = 1 - \theta_{YC}^2$ since $1 = Var(Y) = Var(\theta_{YC} \, C + U_Y) = \theta_{YC}^2 Var(C) + Var(U_Y) = \theta_{YC}^2 + Var(U_Y)$.)
\begin{figure}[!h]
{\footnotesize
$$
\xymatrix@-1.4pc{
(a) & & U_X \ar[dll] \ar[drr] & && (b) & & U_X \ar[dll] \ar[drr] & &\\
*+[F-:<10pt>]{X} & & U_C \ar[dl] & & *+[F-:<10pt>]{X^\ast} & *+[F-:<10pt>]{X} & & U_C \ar[dl] \ar[dr] & & *+[F-:<10pt>]{X^\ast} \\
& *+[F-:<10pt>]{C} \ar[ul] \ar[dl] &  & *+[F-]{c} \ar[rd] \ar[ru] && & *+[F-:<10pt>]{C} \ar[ul] \ar[dl] & & *+[F-:<10pt>]{C} \ar[ru] & \\
*+[F-:<10pt>]{Y} \ar[uu] & & U_Y \ar[ll] \ar[rr] & & *+[F-:<10pt>]{Y^\ast} \ar[uu] & *+[F-:<10pt>]{Y} \ar[uu] & & U_Y \ar[ll] \ar[rr] & & *+[F-:<10pt>]{Y^\ast} \ar[uu] \\
}
$$}
  \caption{Alternative model modifications for the confounding only examples.}
  \label{fig:confoundingtwin2}
\end{figure}
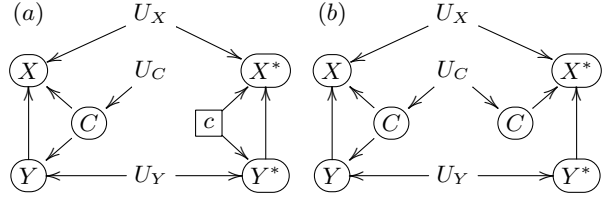
Now, consider another alternative modification where we drop the causal link $C \rightarrow Y$ (rather than $C \rightarrow X$) as shown in Figure \ref{fig:confoundingtwin2}b. Note that direct calculation of $Cov(X^\ast, Y^\ast)$ shows again that,
\begin{align*}
Cov(X^\ast, Y^\ast) &= Cov(\theta_{XY} \, Y^\ast + \theta_{XC} \, C + U_X, Y^\ast) \\
&= \theta_{XY} \, Cov(Y^\ast, Y^\ast) \\
&= \theta_{XY} \, Var(Y^\ast) = \theta_{XY} \, Var(U_Y) \\
&= \theta_{XY} (1 - \theta_{YC}^2)~.
\end{align*}
Hence, we see that for both alternative modifications presented in Figure \ref{fig:confoundingtwin2} the covariance between the response and the feature does not equal $\theta_{XY}$, the association due to the causal effect of $Y$ on $X$. (Note that in both examples the intervention altered the original variable $Y$.)

Now, we show that for the mediation only example, these alternative modifications still capture the correct covariance because, in this case, these modifications do not alter $Y$. For instance, by setting the mediator $M$ to the fixed value $m$, as described in Figure \ref{fig:mediationtwin2}a, we still have that,
\begin{align*}
Cov(X^\ast, Y) &= Cov(\theta_{XY} \, Y + \theta_{XM} \, m + U_X, Y) \\
&= \theta_{XY} \, Cov(Y, Y) + \theta_{XM} \, Cov(m, Y) + Cov(U_X, Y) \\
&= \theta_{XY} \, Var(Y) = \theta_{XY}~.
\end{align*}
Similarly, note that by dropping the causal link $Y \rightarrow M$ (rather than $M \rightarrow X$), as described in Figure \ref{fig:mediationtwin2}b, we still have that,
\begin{align*}
Cov(X^\ast, Y) &= Cov(\theta_{XY} \, Y + \theta_{XM} \, M^\ast + U_X, Y) \\
&= \theta_{XY} \, Cov(Y, Y) + \theta_{XM} \, Cov(M^\ast, Y) + Cov(U_X, Y) \\
&= \theta_{XY} \, Var(Y) = \theta_{XY}~.
\end{align*}
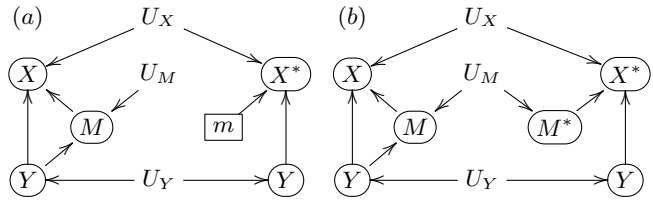
\begin{figure}[!h]
{\footnotesize
$$
\xymatrix@-1.4pc{
(a) & & U_X \ar[dll] \ar[drr] & && (b) & & U_X \ar[dll] \ar[drr] & &\\
*+[F-:<10pt>]{X} & & U_M \ar[dl] & & *+[F-:<10pt>]{X^\ast} & *+[F-:<10pt>]{X} & & U_M \ar[dl] \ar[dr] & & *+[F-:<10pt>]{X^\ast} \\
& *+[F-:<10pt>]{M} \ar[ul] &  & *+[F-]{m} \ar[ru] && & *+[F-:<10pt>]{M} \ar[ul] & & *+[F-:<10pt>]{M^\ast} \ar[ru] & \\
*+[F-:<10pt>]{Y} \ar[ru] \ar[uu] & & U_Y \ar[ll] \ar[rr] & & *+[F-:<10pt>]{Y} \ar[uu] & *+[F-:<10pt>]{Y} \ar[ru] \ar[uu] & & U_Y \ar[ll] \ar[rr] & & *+[F-:<10pt>]{Y} \ar[uu] \\
}
$$}
\vskip -0.1in
  \caption{Alternative model modifications for the for the mediation only examples.}
  \label{fig:mediationtwin2}
\end{figure}

These examples show that for the mediation problem we don't necessarily need to simulate counterfactual features by dropping $M$ from the parent set of $X$. From a practical point of view, however, it is still more advantageous to simulate counterfactual features by dropping the causal link $M \rightarrow X$ since this approach only requires the simulation of the counterfactual features, whereas the approach described in Figure \ref{fig:mediationtwin2}a requires us to set $M$ to $m$, and the approach in Figure \ref{fig:mediationtwin2}b requires the simulation of counterfactual mediator data, $M^\ast$, in addition to the simulation of counterfactual feature data, $X^\ast$.

Now, let's consider indirect causal effects in anticausal prediction tasks. Here, the goal is to simulate counterfactual data where $Cov(X^\ast, Y^\ast) = \theta_{XM} \, \theta_{MY}$. Consider first the alternative intervention where we remove the link $C \rightarrow Y$ (rather than $C \rightarrow X$, as we did in Figure 2 in the main text) in addition to removing $Y \rightarrow X$ and $C \rightarrow M$, as shown in Figure \ref{fig:indirecttwin2}a. Note that, in this case, the intervention altered $Y$ and we have that,
\begin{align*}
Cov(X^\ast, Y^\ast) &= Cov(\theta_{XC} C + \theta_{XM} M^\ast + U_X, U_Y) \\
&= \theta_{XM} \, Cov(M^\ast, U_Y) \\
&= \theta_{XM} \, Cov(\theta_{MY} Y^\ast + U_M, U_Y) \\
&= \theta_{XM} \, \theta_{MY} \, Cov(Y^\ast, U_Y) = \theta_{XM} \, \theta_{MY} \, Var(U_Y) \\
&= \theta_{XM} \, \theta_{MY} \, (1 - \theta_{YC}^2)~.
\end{align*}
Similarly, for the intervention where we set $C$ to $c$ we also alter $Y$ and we have that,
\begin{align*}
Cov(X^\ast, Y^\ast) &= Cov(\theta_{XC} c + \theta_{XM} M^\ast + U_X, Y^\ast) \\
&= \theta_{XM} \, Cov(M^\ast, Y^\ast) \\
&= \theta_{XM} \, Cov(\theta_{MY} Y^\ast + \theta_{MC} \, c + U_M, Y^\ast) \\
&= \theta_{XM} \, \theta_{MY} \, Var(Y^\ast) \\
&= \theta_{XM} \, \theta_{MY} \, Var(\theta_{YC} \, c + U_Y) \\
&= \theta_{XM} \, \theta_{MY} \, Var(U_Y) \\
&= \theta_{XM} \, \theta_{MY} \, (1 - \theta_{YC}^2)~.
\end{align*}
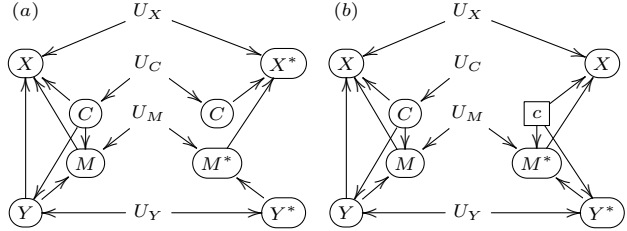
\begin{figure}[!h]
{\scriptsize
$$
\xymatrix@-1.4pc{
(a) & & U_X \ar[dll] \ar[drr] & && (b) & & U_X \ar[dll] \ar[drr] & &\\
*+[F-:<10pt>]{X} & & U_C \ar[dl] \ar[dr] & & *+[F-:<10pt>]{X^\ast} & *+[F-:<10pt>]{X} & & U_C \ar[dl] & & *+[F-:<10pt>]{X} \\
& *+[F-:<10pt>]{C} \ar[ul] \ar[ddl] \ar[d] & U_M \ar[dl] \ar[dr] & *+[F-:<10pt>]{C} \ar[ru] && & *+[F-:<10pt>]{C} \ar[ul] \ar[ddl] \ar[d] & U_M \ar[dl] \ar[dr] & *+[F]{c} \ar[ru] \ar[d] \ar[ddr] & \\
& *+[F-:<10pt>]{M} \ar[uul] && *+[F-:<10pt>]{M^\ast} \ar[uur] &&& *+[F-:<10pt>]{M} \ar[uul] && *+[F-:<10pt>]{M^\ast} \ar[uur] & \\
*+[F-:<10pt>]{Y} \ar[uuu] \ar[ur] & & U_Y \ar[ll] \ar[rr] & & *+[F-:<10pt>]{Y^\ast} \ar[ul] & *+[F-:<10pt>]{Y} \ar[uuu] \ar[ur] & & U_Y \ar[ll] \ar[rr] & & *+[F-:<10pt>]{Y^\ast} \ar[ul] \\
}
$$}
\vskip -0.2in
  \caption{Twin network approach in the case where the indirect effect represents the causal effect of interest.}
  \label{fig:indirecttwin2}
\end{figure}

These examples once again illustrate that we are unable to recover the associations generated by the indirect effects (namely, $\theta_{XM} \, \theta_{MY}$) when we alter $Y$ in anticausal tasks.

\section{Node-splitting transformations as alternative interventions}

In this section we show that the adoption of node-splitting transformations~\cite{richardson2013} encoded in single world intervention graphs (SWIGs) can also be used as an alternative intervention for the generation of counterfactual data that contains only the associations generated by the causal mechanisms of interest. Here, we present SWIGs that capture exactly the same marginal associations between the counterfactual features and responses, as the twin-networks presented in Figure 2 in the main text, and in Supplementary Figures \ref{fig:indirect.confounding.twin}a and b.

Figure \ref{fig:anticausal.node.splitting} presents the SWIGs for the generation of counterfactual features in the anticausal prediction tasks. Here, a node-split operation associated with the intervention $do(Z = z)$ is represented by splitting the node $\xymatrix{*+[F-:<10pt>]{Z}}$ into two elements: $\xymatrix{*+[F]{z}}$ representing the instantiation of $Z$ to the fixed value $z$; and $\xymatrix{*+[F-:<10pt>]{Z}}$ representing the random variable $Z$.

\begin{figure}[!h]
{\scriptsize
$$
\xymatrix@-1.1pc{
(a) & *+[F]{c} \ar[dl] &  *+[F-:<10pt>]{C} \ar[dd] \ar[dr] & & (b) & *+[F]{c} \ar[dl] \ar[ddr] & *+[F-:<10pt>]{C} \ar[drr] & & & (c) & & *+[F-:<10pt>]{C} \ar[dd] \ar[drr] \ar[dll] & & \\
*+[F-:<10pt>]{X_{c,m,Y}} & & & *+[F-:<10pt>]{Y} \ar[lll] \ar[dl] & *+[F-:<10pt>]{X_{y,c,{M_{c,Y}}}} & & & *+[F]{y} \ar[lll]  & *+[F-:<10pt>]{Y} \ar[dll] & *+[F-:<10pt>]{X_{y,m,C}} & & & *+[F]{y} \ar[lll]  & *+[F-:<10pt>]{Y} \ar[dll] \\
& *+[F]{m} \ar[ul] & *+[F-:<10pt>]{M} & & & & *+[F-:<10pt>]{M_{c,Y}} \ar[ull] & & & & *+[F]{m} \ar[ul] & *+[F-:<10pt>]{M} & & \\
}
$$}
\vskip -0.1in
  \caption{SWIGs for the anticausal predictive tasks.}
  \label{fig:anticausal.node.splitting}
\end{figure}
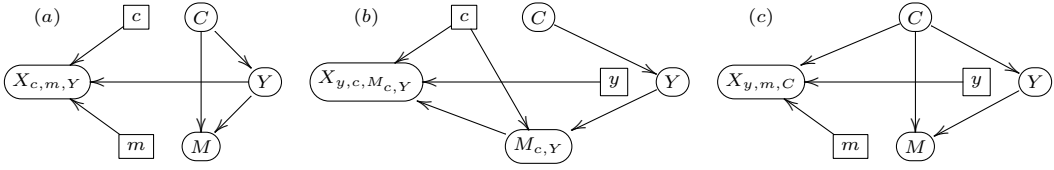

In Figure \ref{fig:anticausal.node.splitting}a we split the $C$ and $M$ nodes in order to obtain a counterfactual feature $X_{c,m,Y}$, whose association with $Y$ is generated by the direct causal effect $\theta_{XY}$, since for any fixed values of $c$ and $m$ we have that,
\begin{align*}
Cov(X_{c,m,Y}, Y) &= Cov(\theta_{XC} \, c + \theta_{XM} \, m + \theta_{XY} \, Y + U_X, Y) \\
&= \theta_{XY} \, Cov(Y, Y) = \theta_{XY}~.
\end{align*}
In Figure \ref{fig:anticausal.node.splitting}b we split the $C$ and $Y$ nodes in order to obtain a counterfactual feature $X_{y,c,{M_{c,Y}}}$, whose association with $Y$ is generated by the indirect causal effect $\theta_{XM} \, \theta_{MY}$, since for any fixed values of $c$ and $y$ we have that,
\begin{align*}
Cov(X_{y,c,{M_{c,Y}}}, Y) &= Cov(\theta_{XC} \, c + \theta_{XY} \, y + \theta_{XM} \, M_{c,Y} + U_X, Y) \\
&= \theta_{XM} \, Cov(M_{c,Y}, Y) \\
&= \theta_{XM} \, Cov(\theta_{MC} \, c + \theta_{MY} \, Y + U_M, Y) \\
&= \theta_{XM} \, \theta_{MY} \, Cov(Y, Y) = \theta_{XM} \, \theta_{MY}~.
\end{align*}
Finally, in Figure \ref{fig:anticausal.node.splitting}c we split the $Y$ and $M$ nodes in order to obtain a counterfactual feature $X_{y,m,C}$, whose association with $Y$, measured by $\theta_{XC} \, \theta_{YC}$, is generated by the confounder $C$. Note that for any fixed values of $y$ and $m$ we have that,
\begin{align*}
Cov(X_{y,m,C}, Y) &= Cov(\theta_{XC} \, C + \theta_{XM} \, m + \theta_{XY} \, y + U_X, Y) \\
&= \theta_{XC} \, Cov(C, Y) \\
&= \theta_{XC} \, Cov(C, \theta_{YC} \, C + U_Y) \\
&= \theta_{XC} \, \theta_{YC} \, Cov(C, C) = \theta_{XC} \, \theta_{YC}~.
\end{align*}

Note that in the SWIG framework, even when we split the $Y$ node into $\xymatrix{*+[F]{y}}$ and $\xymatrix{*+[F-:<10pt>]{Y}}$ in anticausal prediction tasks (e.g., Figure \ref{fig:anticausal.node.splitting}b and c), we have that the component $\xymatrix{*+[F-:<10pt>]{Y}}$ still represents the un-altered random variable $Y$. (This observation is again consistent with the point made in the previous section that for anticausal prediction tasks, the requirement for the intervention to work is that $Y$ is not altered by the intervention.)

\section{Anticausal reparameterization example}

Here, we present an illustrative example of the reparameterization for the anticausal prediction task. The goal is to provide a concrete example to help out readers that are not familiar with the notation used in the linear structural equations models. Figure \ref{fig:anticausal.example}a presents an illustrative example of the actual data generation process, whereas Figure \ref{fig:anticausal.example.reparameterized} represents the reparameterized model.


\begin{figure}[!h]
{\scriptsize
$$
\xymatrix@-1.0pc{
& (a) & & U_{C_1} \ar[r] & *+[F-:<10pt>]{C_1} \ar[rd] &&& (b) & *+[F-:<10pt>]{\bfC} \ar[dl] \ar[dr] \ar[dd] & & (\bfX) & (\bfC) & (\bfM) \\
& & U_{C_2} \ar[r] & *+[F-:<10pt>]{C_2} \ar[rr] \ar[dl] \ar[ru] & & *+[F-:<10pt>]{C_3} \ar[ddr] \ar[dddd] & U_{C_3} \ar[l] & *+[F-:<10pt>]{\bfX} & & *+[F-:<10pt>]{Y} \ar[ll] \ar[dl] & *+[F-:<10pt>]{X_1} \ar[d] \ar@/^1.0pc/[dd] & *+[F-:<10pt>]{C_1} \ar@/^1.0pc/[dd] & *+[F-:<10pt>]{M_1} \\
& U_{X_1} \ar[r] & *+[F-:<10pt>]{X_1} \ar[dl] \ar[dd] &&&&&& *+[F-:<10pt>]{\bfM} \ar[ul] && *+[F-:<10pt>]{X_2} \ar[d] & *+[F-:<10pt>]{C_2} \ar[u] \ar[d] & *+[F-:<10pt>]{M_2} \ar[u] \\
U_{X_2} \ar[r] & *+[F-:<10pt>]{X_2} \ar[dr]_{\theta_{{X_3}{X_2}}} & & & & & *+[F-:<10pt>]{Y} \ar[lllll]_{\theta_{{X_2}Y}} \ar[dllll]^{\theta_{{X_3}Y}} \ar[ddl] & U_{Y} \ar[l] &&& *+[F-:<10pt>]{X_3} & *+[F-:<10pt>]{C_3} \\
& & *+[F-:<10pt>]{X_3} &&&&&&&& (c) & (d) & (e) \\
& U_{X_3} \ar[ru] & U_{M_1} \ar[r] & *+[F-:<10pt>]{M_1} \ar[lu] & & *+[F-:<10pt>]{M_2} \ar[ll] \ar@/^1pc/[uuulll] & U_{M_2} \ar[l] & \\
}
$$}
\vskip -0.1in
  \caption{Original anticausal prediction task example. Panel a shows the actual data generation process. Panel b shows the multivariate representation of the DAG in panel a. Panels c, d, and e show, respectively, the DAG subdiagrams represented by the $\bfX$, $\bfC$, and $\bfM$ nodes in panel b.}
  \label{fig:anticausal.example}
  \vskip -0.1in
\end{figure}
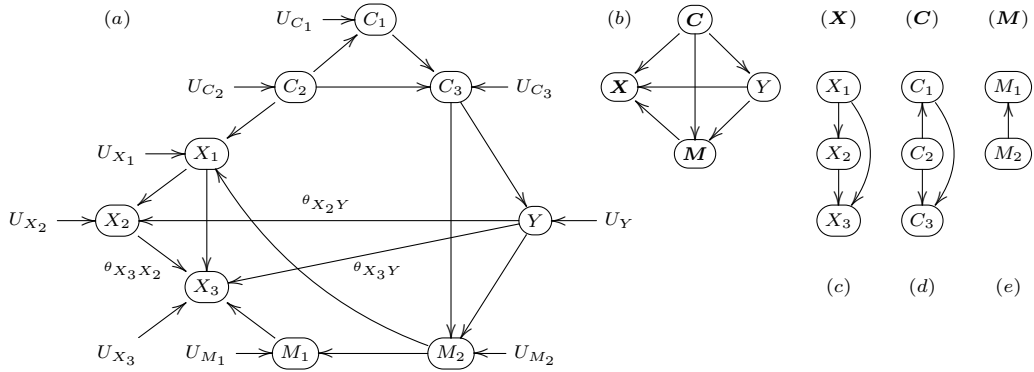

\begin{figure}[!h]
{\scriptsize
$$
\xymatrix@-1.2pc{
&&& U_{C_1} \ar[d] \ar[drr] & U_{C_2} \ar[d] \ar[dl] \ar[dr] & U_{C_3} \ar[d] && \\
&&& W_{C_1} \ar[d] & W_{C_2} \ar[d] & W_{C_3} \ar[d] && \\
& & & *+[F-:<10pt>]{C_1} & *+[F-:<10pt>]{C_2} \ar[ddll] \ar[dddll] \ar[ddddll] & *+[F-:<10pt>]{C_3} \ar[dddr] \ar[dddddd] \ar@/_1.0pc/[ddddddll] & \\
&&&&&& U_Y \ar[d] \\
U_{X_1} \ar[r] \ar[dr] \ar[ddr] & W_{X_1} \ar[r] & *+[F-:<10pt>]{X_1} &&&&  W_Y \ar[d] \\
U_{X_2} \ar[r] \ar[dr] & W_{X_2} \ar[r] & *+[F-:<10pt>]{X_2} & & & & *+[F-:<10pt>]{Y} \ar[llll]|-{\gamma_{{X_2}Y}} \ar[dllll]|-{\gamma_{{X_3}Y}} \ar[dddl] \ar[dddlll] \\
U_{X_3} \ar[r] & W_{X_3} \ar[r] & *+[F-:<10pt>]{X_3} &&&& \\
&&&&& \\
&& & *+[F-:<10pt>]{M_1} \ar[luu] & & *+[F-:<10pt>]{M_2} \ar@/^0.5pc/[uuuulll] \ar@/^0.5pc/[uuulll] \ar[uulll] & & \\
&&& W_{M_1} \ar[u] && W_{M_2} \ar[u] && \\
&&& U_{M_1} \ar[u] && U_{M_2} \ar[u] \ar[ull] && \\
}
$$}
\vskip -0.1in
  \caption{Reparameterized model for the anticausal prediction task example in Figure \ref{fig:anticausal.example}a.}
  \label{fig:anticausal.example.reparameterized}
\end{figure}

For the anticausal prediction task DAG in Figure \ref{fig:anticausal.example}, we have that the structural equations,
\begin{align*}
\bfC &= \bfTheta_{CC} \, \bfC + \bfU_C~, \\
Y &= \bfTheta_{YC} \, \bfC + U_Y~, \\
\bfM &= \bfTheta_{MM} \, \bfM + \bfTheta_{MC} \, \bfC + \bfTheta_{MY} \, Y + \bfU_M~, \\
\bfX &= \bfTheta_{XX} \, \bfX + \bfTheta_{XC} \, \bfC + \bfTheta_{XM} \, \bfM + \bfTheta_{XY} \, Y + \bfU_X~,
\end{align*}
are explicitly given by,
\begin{equation*}
\underbrace{
\begin{pmatrix}
C_1 \\
C_2 \\
C_3 \\
\end{pmatrix}}_{\bfC}
=
\underbrace{
\begin{pmatrix}
0 & \theta_{{C_1}{C_2}} & 0 \\
0 & 0 & 0 \\
\theta_{{C_3}{C_1}} & \theta_{{C_3}{C_2}} & 0 \\
\end{pmatrix}}_{\bfTheta_{CC}}
\underbrace{
\begin{pmatrix}
C_1 \\
C_2 \\
C_3 \\
\end{pmatrix}}_{\bfC}
+
\underbrace{
\begin{pmatrix}
U_{C_1} \\
U_{C_2} \\
U_{C_3} \\
\end{pmatrix}}_{\bfU_C}~,
\end{equation*}

\begin{equation*}
Y =
\underbrace{
\begin{pmatrix}
0 & 0 & \theta_{{Y}{C_3}} \\
\end{pmatrix}}_{\bfTheta_{YC}}
\underbrace{
\begin{pmatrix}
C_1 \\
C_2 \\
C_3 \\
\end{pmatrix}}_{\bfC}
+ \, U_Y~,
\end{equation*}

\begin{equation*}
\underbrace{
\begin{pmatrix}
M_1 \\
M_2 \\
\end{pmatrix}}_{\bfM}
=
\underbrace{
\begin{pmatrix}
0 & \theta_{{M_1}{M_2}} \\
0 & 0 \\
\end{pmatrix}}_{\bfTheta_{MM}}
\underbrace{
\begin{pmatrix}
M_1 \\
M_2 \\
\end{pmatrix}}_{\bfM} +
\underbrace{
\begin{pmatrix}
0 & 0 & 0 \\
0 & 0 & \theta_{{M_2}{C_3}} \\
\end{pmatrix}}_{\bfTheta_{MC}}
\underbrace{
\begin{pmatrix}
C_1 \\
C_2 \\
C_3 \\
\end{pmatrix}}_{\bfC} +
\underbrace{
\begin{pmatrix}
0 \\
\theta_{{M_2}{Y}} \\
\end{pmatrix}}_{\bfTheta_{MY}}
Y
+
\underbrace{
\begin{pmatrix}
U_{M_1} \\
U_{M_2} \\
\end{pmatrix}}_{\bfU_M}~,
\end{equation*}

\begin{align*}
\underbrace{
\begin{pmatrix}
X_1 \\
X_2 \\
X_3 \\
\end{pmatrix}}_{\bfX}
=&
\underbrace{
\begin{pmatrix}
0 & 0 & 0 \\
\theta_{{X_2}{X_1}} & 0 & 0 \\
\theta_{{X_3}{X_1}} & \theta_{{X_3}{X_2}} & 0 \\
\end{pmatrix}}_{\bfTheta_{XX}}
\underbrace{
\begin{pmatrix}
X_1 \\
X_2 \\
X_3 \\
\end{pmatrix}}_{\bfX} +
\underbrace{
\begin{pmatrix}
0 & \theta_{{X_1}{C_2}} & 0 \\
0 & 0 & 0 \\
0 & 0 & 0 \\
\end{pmatrix}}_{\bfTheta_{XC}}
\underbrace{
\begin{pmatrix}
C_1 \\
C_2 \\
C_3 \\
\end{pmatrix}}_{\bfC} + \\
&+
\underbrace{
\begin{pmatrix}
0 & \theta_{{X_1}{M_2}} \\
0 & 0 \\
\theta_{{X_3}{M_1}} & 0 \\
\end{pmatrix}}_{\bfTheta_{XM}}
\underbrace{
\begin{pmatrix}
M_1 \\
M_2 \\
\end{pmatrix}}_{\bfM} +
\underbrace{
\begin{pmatrix}
0 \\
\theta_{{X_2}{Y}} \\
\theta_{{X_3}{Y}} \\
\end{pmatrix}}_{\bfTheta_{XY}}
Y
+
\underbrace{
\begin{pmatrix}
U_{X_1} \\
U_{X_2} \\
U_{X_3} \\
\end{pmatrix}}_{\bfU_X}~.
\end{align*}

Using simple algebraic manipulations, we can re-write the above linear structural models as,
\begin{align*}
\bfC &= \bfW_C~, \\
Y &= \bfGamma_{YC} \, \bfC + W_Y~, \\
\bfM &= \bfGamma_{MC} \, \bfC + \bfGamma_{MY} \, Y + \bfW_M~, \\
\bfX &= \bfGamma_{XC} \, \bfC + \bfGamma_{XM} \, \bfM + \bfGamma_{XY} \, Y + \bfW_X~,
\end{align*}
where,
\begin{align*}
\bfW_C = (\bfI - \bfTheta_{CC})^{-1} \bfU_C~, \; W_Y = U_Y~, \; \bfW_M = (\bfI - \bfTheta_{MM})^{-1} \bfU_M~, \; \bfW_X = (\bfI - \bfTheta_{XX})^{-1} \bfU_X~,
\end{align*}
and,
\begin{align*}
\bfGamma_{YC} &= \bfTheta_{YC}~, \;\; \bfGamma_{MC} = (\bfI - \bfTheta_{MM})^{-1} \bfTheta_{MC}~, \;\; \bfGamma_{MY} = (\bfI - \bfTheta_{MM})^{-1} \bfTheta_{MY}~, \\
\bfGamma_{XC} &= (\bfI - \bfTheta_{XX})^{-1} \bfTheta_{XC}~, \;\; \bfGamma_{XM} = (\bfI - \bfTheta_{XX})^{-1} \bfTheta_{XM}~, \;\; \bfGamma_{XY} = (\bfI - \bfTheta_{XX})^{-1} \bfTheta_{XY}~.
\end{align*}

Next, the present the explicit form of parameters and error terms for the particular example in Figure \ref{fig:anticausal.example}. Starting with model $\bfC = \bfW_C$, we have that,
\begin{equation*}
(\bfI - \bfTheta_{CC})^{-1} =
\begin{pmatrix}
1 & \theta_{{C_1}{C_2}} & 0 \\
0 & 1 & 0 \\
\theta_{{C_3}{C_1}} & \theta_{{C_3}{C_2}} + \theta_{{C_3}{C_1}} \, \theta_{{C_1}{C_2}} & 1 \\
\end{pmatrix}~,
\end{equation*}
so that,
\begin{align*}
\bfW_C &= (\bfI - \bfTheta_{CC})^{-1} \bfU_{C} \\
&=
\begin{pmatrix}
1 & \theta_{{C_1}{C_2}} & 0 \\
0 & 1 & 0 \\
\theta_{{C_3}{C_1}} & \theta_{{C_3}{C_2}} + \theta_{{C_3}{C_1}} \, \theta_{{C_1}{C_2}} & 1 \\
\end{pmatrix}
\begin{pmatrix}
U_{C_1} \\
U_{C_2} \\
U_{C_3} \\
\end{pmatrix} \\
&=
\begin{pmatrix}
U_{C_1} + \theta_{{C_1}{C_2}} \, U_{C_2} \\
U_{C_2} \\
U_{C_3} + U_{C_2} (\theta_{{C_3}{C_2}} + \theta_{{C_3}{C_1}} \, \theta_{{C_1}{C_2}}) + U_{C_1} \, \theta_{{C_3}{C_1}} \\
\end{pmatrix} =
\begin{pmatrix}
W_{{C_1}} \\
W_{{C_2}} \\
W_{{C_3}} \\
\end{pmatrix}~,
\end{align*}

For the model $Y = \bfGamma_{YC} \, \bfC + W_Y$, we have that,
\begin{align*}
\bfGamma_{YC} &= \bfTheta_{YC} \\
&=
\begin{pmatrix}
0 & 0 & \theta_{{Y}{C_3}} \\
\end{pmatrix} \\
&=
\begin{pmatrix}
\gamma_{{Y}{C_1}} & \gamma_{{Y}{C_2}} & \gamma_{{Y}{C_3}} \\
\end{pmatrix}~, \\
W_Y &= U_Y~.
\end{align*}

For the model $\bfM = \bfGamma_{MC} \, \bfC + \bfGamma_{MY} \, Y + \bfW_M$, we have that,
\begin{equation*}
(\bfI - \bfTheta_{MM})^{-1} =
\begin{pmatrix}
1 & \theta_{{M_1}{M_2}}\\
0 & 1 \\
\end{pmatrix}~,
\end{equation*}
so that,
\begin{align*}
\bfGamma_{MC} &= (\bfI - \bfTheta_{MM})^{-1} \bfTheta_{MC} \\
&=
\begin{pmatrix}
1 & \theta_{{M_1}{M_2}}\\
0 & 1 \\
\end{pmatrix}
\begin{pmatrix}
0 & 0 & 0 \\
0 & 0 & \theta_{{M_2}{C_3}} \\
\end{pmatrix} \\
&=
\begin{pmatrix}
0 & 0 & \theta_{{M_1}{M_2}} \, \theta_{{M_2}{C_3}} \\
0 & 0 & \theta_{{M_2}{C_3}} \\
\end{pmatrix} =
\begin{pmatrix}
\gamma_{{M_1}{C_1}} & \gamma_{{M_1}{C_2}} & \gamma_{{M_1}{C_3}} \\
\gamma_{{M_2}{C_1}} & \gamma_{{M_2}{C_2}} & \gamma_{{M_2}{C_3}} \\
\end{pmatrix}~,
\end{align*}
and,
\begin{align*}
\bfGamma_{MY} &= (\bfI - \bfTheta_{MM})^{-1} \bfTheta_{MY} \\
&=
\begin{pmatrix}
1 & \theta_{{M_1}{M_2}}\\
0 & 1 \\
\end{pmatrix}
\begin{pmatrix}
0 \\
\theta_{{M_2}{Y}} \\
\end{pmatrix} =
\begin{pmatrix}
\theta_{{M_1}{M_2}} \, \theta_{{M_2}{Y}} \\
\theta_{{M_2}{Y}} \\
\end{pmatrix}
=
\begin{pmatrix}
\gamma_{{M_1}{Y}} \\
\gamma_{{M_2}{Y}} \\
\end{pmatrix}~,
\end{align*}
and,
\begin{align*}
\bfW_{M} &= (\bfI - \bfTheta_{MM})^{-1} \bfU_{M} \\
&=
\begin{pmatrix}
1 & \theta_{{M_1}{M_2}}\\
0 & 1 \\
\end{pmatrix}
\begin{pmatrix}
U_{M_1} \\
U_{M_2} \\
\end{pmatrix} =
\begin{pmatrix}
U_{M_1} + \theta_{{M_1}{M_2}} \, U_{{M_2}} \\
U_{{M_2}} \\
\end{pmatrix}
=
\begin{pmatrix}
W_{{M_1}} \\
W_{{M_2}} \\
\end{pmatrix}~.
\end{align*}

Finally, for the model $\bfX = \bfGamma_{XC} \, \bfC + \bfGamma_{XM} \, \bfM + \bfGamma_{XY} \, Y + \bfW_X$, we have that,
\begin{equation*}
(\bfI - \bfTheta_{XX})^{-1} =
\begin{pmatrix}
1 & 0 & 0 \\
\theta_{{X_2}{X_1}} & 1 & 0 \\
\theta_{{X_3}{X_1}} + \theta_{{X_2}{X_1}} \theta_{{X_3}{X_2}} & \theta_{{X_3}{X_2}} & 1 \\
\end{pmatrix}~,
\end{equation*}
so that,
\begin{align*}
\bfGamma_{XC} &= (\bfI - \bfTheta_{XX})^{-1} \bfTheta_{XC} \\
&=
\begin{pmatrix}
1 & 0 & 0 \\
\theta_{{X_2}{X_1}} & 1 & 0 \\
\theta_{{X_3}{X_1}} + \theta_{{X_2}{X_1}} \, \theta_{{X_3}{X_2}} & \theta_{{X_3}{X_2}} & 1 \\
\end{pmatrix}
\begin{pmatrix}
0 & \theta_{{X_1}{C_2}} & 0 \\
0 & 0 & 0 \\
0 & 0 & 0 \\
\end{pmatrix} \\
&=
\begin{pmatrix}
0 & \theta_{{X_1}{C_2}} & 0 \\
0 & \theta_{{X_1}{C_2}} \, \theta_{{X_2}{X_1}} & 0 \\
0 & \theta_{{X_1}{C_2}} \, \theta_{{X_3}{X_1}} + \theta_{{X_1}{C_2}} \, \theta_{{X_2}{X_1}} \, \theta_{{X_3}{X_2}} & 0 \\
\end{pmatrix} =
\begin{pmatrix}
\gamma_{{X_1}{C_1}} & \gamma_{{X_1}{C_2}} & \gamma_{{X_1}{C_3}}  \\
\gamma_{{X_2}{C_1}} & \gamma_{{X_2}{C_2}} & \gamma_{{X_2}{C_3}}  \\
\gamma_{{X_3}{C_1}} & \gamma_{{X_3}{C_2}} & \gamma_{{X_3}{C_3}}  \\
\end{pmatrix}~,
\end{align*}
and,
\begin{align*}
\bfGamma_{XM} &= (\bfI - \bfTheta_{XX})^{-1} \bfTheta_{XM} \\
&=
\begin{pmatrix}
1 & 0 & 0 \\
\theta_{{X_2}{X_1}} & 1 & 0 \\
\theta_{{X_3}{X_1}} + \theta_{{X_2}{X_1}} \, \theta_{{X_3}{X_2}} & \theta_{{X_3}{X_2}} & 1 \\
\end{pmatrix}
\begin{pmatrix}
0 & \theta_{{X_1}{M_2}} \\
0 & 0 \\
\theta_{{X_3}{M_1}} & 0 \\
\end{pmatrix} \\
&=
\begin{pmatrix}
0 & \theta_{{X_1}{M_2}} \\
0 & \theta_{{X_1}{M_2}} \, \theta_{{X_2}{X_1}} \\
\theta_{{X_3}{M_1}} & \theta_{{X_1}{M_2}} \, \theta_{{X_3}{X_1}} + \theta_{{X_1}{M_2}} \, \theta_{{X_2}{X_1}} \, \theta_{{X_3}{X_2}} \\
\end{pmatrix} =
\begin{pmatrix}
\gamma_{{X_1}{M_1}} & \gamma_{{X_1}{M_2}} \\
\gamma_{{X_2}{M_1}} & \gamma_{{X_2}{M_2}} \\
\gamma_{{X_3}{M_1}} & \gamma_{{X_3}{M_2}} \\
\end{pmatrix}~,
\end{align*}
and,
\begin{align*}
\bfGamma_{XY} &= (\bfI - \bfTheta_{XX})^{-1} \bfTheta_{XY} \\
&=
\begin{pmatrix}
1 & 0 & 0 \\
\theta_{{X_2}{X_1}} & 1 & 0 \\
\theta_{{X_3}{X_1}} + \theta_{{X_2}{X_1}} \theta_{{X_3}{X_2}} & \theta_{{X_3}{X_2}} & 1 \\
\end{pmatrix}
\begin{pmatrix}
0 \\
\theta_{{X_2}{Y}} \\
\theta_{{X_3}{Y}} \\
\end{pmatrix} \\
&=
\begin{pmatrix}
0 \\
\theta_{{X_2}{Y}} \\
\theta_{{X_3}{Y}} + \theta_{{X_3}{X_2}} \, \theta_{{X_2}{Y}} \\
\end{pmatrix}
=
\begin{pmatrix}
\gamma_{{X_1}{Y}} \\
\gamma_{{X_2}{Y}} \\
\gamma_{{X_3}{Y}} \\
\end{pmatrix}~,
\end{align*}
and,
\begin{align*}
\bfW_X &= (\bfI - \bfTheta_{XX})^{-1} \bfU_{X} \\
&=
\begin{pmatrix}
1 & 0 & 0 \\
\theta_{{X_2}{X_1}} & 1 & 0 \\
\theta_{{X_3}{X_1}} + \theta_{{X_2}{X_1}} \theta_{{X_3}{X_2}} & \theta_{{X_3}{X_2}} & 1 \\
\end{pmatrix}
\begin{pmatrix}
U_{X_1} \\
U_{X_2} \\
U_{X_3} \\
\end{pmatrix} \\
&=
\begin{pmatrix}
U_{X_1} \\
U_{X_1} \theta_{{X_2}{X_1}} + U_{X_2} \\
U_{X_1}(\theta_{{X_3}{X_1}} + \theta_{{X_2}{X_1}} \theta_{{X_3}{X_2}}) + U_{X_2} \theta_{{X_3}{X_2}} + U_{X_3} \\
\end{pmatrix} =
\begin{pmatrix}
W_{{U_1}} \\
W_{{U_2}} \\
W_{{U_3}} \\
\end{pmatrix}~.
\end{align*}


Table \ref{tab:table1} compiles all the elements of $\bfGamma_{YC}$, $\bfGamma_{MC}$, $\bfGamma_{MY}$, $\bfGamma_{XC}$, $\bfGamma_{XM}$, and $\bfGamma_{XY}$. It presents the causal effects in the reparameterized model (represented by the $\gamma$s) in terms of the original causal effects (represented by the $\theta$s). Note that the arrows in Figure \ref{fig:anticausal.example.reparameterized} correspond to the non-zero $\gamma$ causal effects in Table \ref{tab:table1}.
\begin{table}[!h]
\begin{center}
\begin{tabular}{c|l}
\hline
$\gamma \, \in$ & $\gamma_{{Y}{C_3}} = \theta_{{Y}{C_3}}$ \\
$\bfGamma_{YC}$ & $\gamma_{{Y}{C_1}} = \gamma_{{Y}{C_2}} = 0$ \\
\hline
$\gamma \, \in$ & $\gamma_{{M_1}{C_3}} = \theta_{{M_1}{M_2}} \, \theta_{{M_2}{C_3}}$ \\
$\bfGamma_{MC}$ & $\gamma_{{M_2}{C_3}} = \theta_{{M_2}{C_3}}$ \\
& $\gamma_{{M_1}{C_1}} = \gamma_{{M_1}{C_2}} = \gamma_{{M_2}{C_1}} = \gamma_{{M_2}{C_2}} = 0$ \\
\hline
$\gamma \, \in$ & $\gamma_{{M_1}Y} = \theta_{{M_1}{M_2}} \, \theta_{{M_2}{Y}}$ \\
$\bfGamma_{MY}$ & $\gamma_{{M_2}Y} = \theta_{{M_2}{Y}}$ \\
\hline
$\gamma \, \in$ & $\gamma_{{X_1}{C_2}} = \theta_{{X_1}{C_2}}$ \\
$\bfGamma_{XC}$ & $\gamma_{{X_2}{C_2}} = \theta_{{X_2}{X_1}} \, \theta_{{X_1}{C_2}}$ \\
& $\gamma_{{X_3}{C_2}} = \theta_{{X_1}{C_2}} \, \theta_{{X_3}{X_1}} + \theta_{{X_1}{C_2}} \, \theta_{{X_2}{X_1}} \, \theta_{{X_3}{X_2}}$ \\
& $\gamma_{{X_1}{C_1}} = \gamma_{{X_1}{C_3}} = \gamma_{{X_2}{C_1}} = \gamma_{{X_2}{C_3}} = $ \\
& $= \gamma_{{X_3}{C_1}} = \gamma_{{X_3}{C_3}} = 0$ \\
\hline
$\gamma \, \in$ & $\gamma_{{X_3}{M_1}} = \theta_{{X_3}{M_1}}$ \\
$\bfGamma_{XM}$ & $\gamma_{{X_1}{M_2}} = \theta_{{X_1}{M_2}}$ \\
& $\gamma_{{X_2}{M_2}} = \theta_{{X_1}{M_2}} \, \theta_{{X_2}{X_1}}$ \\
& $\gamma_{{X_3}{M_2}} = \theta_{{X_1}{M_2}} \, \theta_{{X_3}{X_1}} + \theta_{{X_1}{M_2}} \, \theta_{{X_2}{X_1}} \, \theta_{{X_3}{X_2}}$ \\
& $\gamma_{{X_1}{M_1}} = \gamma_{{X_2}{M_1}} = 0$ \\
\hline
$\gamma \, \in$ & $\gamma_{{X_2}Y} = \theta_{{X_2}{Y}}$ \\
$\bfGamma_{XY}$ & $\gamma_{{X_3}Y} = \theta_{{X_3}{Y}} + \theta_{{X_3}{X_2}} \, \theta_{{X_2}{Y}}$ \\
& $\gamma_{{X_1}Y} = 0$ \\
\hline
\end{tabular}
\end{center}
\caption{Causal effects in the reparameterized model.}
\label{tab:table1}
\vskip -0.1in
\end{table}

\subsection{On the invertibility of $(\bfI - \bfTheta_{VV})$}

Here, it is important to point out that for any arbitrary DAG, we have that the matrix $(\bfI - \bfTheta_{VV})$ is always invertible. To see why this is the case, note that for any arbitrary DAG we can always rearrange the order of the variables so that $\bfTheta_{VV}$ is a lower triangular matrix. For instance, we can rename and rearrange the order of the variables in the DAG $M_a$ in (\ref{fig:order.rearrangement}) as $V'_4 = V_1$, $V'_1 = V_2$, $V'_3 = V_3$, and $V'_2 = V_4$, in order to obtain the rearranged DAG $M_b$ in (\ref{fig:order.rearrangement}).
\begin{equation}
\xymatrix@-0pc{
(M_a) & *+[F-:<10pt>]{V'_4} \ar[dd]|-{\theta_{{V'_3}{V'_4}}} \ar[rr]|-{\theta_{{V'_1}{V'_4}}} \ar[ddrr]|-{\theta_{{V'_2}{V'_4}}} && *+[F-:<10pt>]{V'_1} \ar[dd]|-{\theta_{{V'_2}{V'_1}}} && (M_b) & *+[F-:<10pt>]{V_1} \ar[dd]|-{\theta_{{V_3}{V_1}}} \ar[rr]|-{\theta_{{V_2}{V_1}}} \ar[ddrr]|-{\theta_{{V_4}{V_1}}} && *+[F-:<10pt>]{V_2} \ar[dd]|-{\theta_{{V_4}{V_2}}} \\
&&& && &&& \\
& *+[F-:<10pt>]{V'_3} \ar[rr]|-{\theta_{{V'_2}{V'_3}}} && *+[F-:<10pt>]{V'_2} & & & *+[F-:<10pt>]{V_3} \ar[rr]|-{\theta_{{V_4}{V_3}}} && *+[F-:<10pt>]{V_4}\\
}
\label{fig:order.rearrangement}
\end{equation}

\noindent In this way, the original matrix $\bfTheta_{V'V'}$,
\begin{equation}
\bfTheta_{V'V'} =
\begin{pmatrix}
0 & 0 & 0 & \theta_{{V'_1}{V'_4}} \\
\theta_{{V'_2}{V'_1}} & 0 & \theta_{{V'_2}{V'_3}} & \theta_{{V'_2}{V'_4}} \\
0 & 0 & 0 & \theta_{{V'_3}{V'_4}} \\
0 & 0 & 0 & 0 \\
\end{pmatrix}~,
\end{equation}
is rearranged as the lower triangular matrix,
\begin{equation}
\bfTheta_{VV} =
\begin{pmatrix}
0 & 0 & 0 & 0 \\
\theta_{{V_2}{V_1}} & 0 & 0 & 0 \\
\theta_{{V_3}{V_1}} & 0 & 0 & 0 \\
\theta_{{V_4}{V_1}} & \theta_{{V_4}{V_2}} & \theta_{{V_4}{V_3}} & 0 \\
\end{pmatrix}~.
\end{equation}
Now, recalling that the determinant of a (lower or upper) triangular matrix is given by the product of its diagonal elements and that a triangular matrix is invertible if and only if none of its diagonal elements is zero, we see that $(\bfI - \bfTheta_{VV})$ is always invertible because all diagonal elements are always equal to 1.

\section{Remarks on identification issues}

Under the assumption that all the confounders and mediators are observed, we can identify the direct and indirect causal effects of response on the features. In particular, a simple least squares estimation procedure provides consistent estimates of these causal effects\footnote{Here, we assume that the number of samples is larger than the number of covariates in the regression fits, and that multicolinearity is not an issue too.}. To see why, note that for the reparameterized model, if all confounders and mediators are observed, it follows from the Markov property of DAGs that $X_j = f_{X_j}(\bfC, \bfM, Y, W_{X_j}) = f_{X_j}(pa(X_j), W_{X_j})$. (Here, $f_{X_j}$ represents linear structural causal models). Hence, for the anticausal task, it follows that, when we regress $X_j$ on the elements of $\bfC$, $\bfM$, and $Y$ only the coefficients associated with the parents of $X_j$ in the reparameterized model will be statistically different from zero (for large enough sample sizes). Therefore, in practice, we don't need to know before hand which variables are the parents of $X_j$ in the reparameterized model. The parent set will be learned automatically from the data by the regression model fit.

Observe, as well, that even if the mediators are unobserved, but the confounders are still observed, we can still identify total causal effects. For instance, we have that,
\begin{align*}
\bfX &= \bfGamma_{XC} \, \bfC + \bfGamma_{XM} \, \bfM + \bfGamma_{XY} \, Y + \bfW_X~, \\
&= \bfGamma_{XC} \, \bfC + \bfGamma_{XM} \, (\bfGamma_{MC} \, \bfC + \bfGamma_{MY} \, Y + \bfW_M) + \bfGamma_{XY} \, Y + \bfW_X~, \\
&= \underbrace{(\bfGamma_{XC} + \bfGamma_{XM} \, \bfGamma_{MC})}_{\tilde{\bfGamma}_{XC}} \, \bfC + \underbrace{(\bfGamma_{XY} + \bfGamma_{XM} \, \bfGamma_{MY})}_{\tilde{\bfGamma}_{XY}} \, Y + \underbrace{\bfGamma_{XM} \, \bfW_M+ \bfW_X}_{\tilde{\bfW}_X}~, \\
&= \tilde{\bfGamma}_{XC} \, \bfC + \tilde{\bfGamma}_{XY} \, Y + \tilde{\bfW}_X~,
\end{align*}
where $\tilde{\bfGamma}_{XY} = \bfGamma_{XY} + \bfGamma_{XM} \, \bfGamma_{MY}$ represents the total causal effect of $Y$ on $\bfX$, as represented in the DAG $M_b$ in in the causal task model (\ref{eq:anticausal.task.total}).
\begin{equation}
\xymatrix@-0.0pc{
(M_a) & *+[F-:<10pt>]{\bfC} \ar[dl]_{\bfGamma_{XC}} \ar[dr]^{\bfGamma_{YC}} \ar@/^1.5pc/[dd]|-{\bfGamma_{MC}} & & (M_b) & *+[F-:<10pt>]{\bfC} \ar[dl]_{\tilde{\bfGamma}_{XC}} \ar[dr]^{\bfGamma_{YC}} & \\
*+[F-:<10pt>]{\bfX} & & *+[F-:<10pt>]{Y} \ar@/^0.75pc/[ll]|-{\bfGamma_{XY}} \ar[dl]^{\bfGamma_{MY}}  & *+[F-:<10pt>]{\bfX} && *+[F-:<10pt>]{Y} \ar[ll]^{\tilde{\bfGamma}_{XY}} \\
& *+[F-:<10pt>]{\bfM} \ar[ul]^{\bfGamma_{XM}} & & &  & \\
}
\label{eq:anticausal.task.total}
\end{equation}

On the other hand, if the mediators are observed, but some the confounders are unobserved, then neither the direct, the indirect, or the total causal effects are identifiable, and the predictions generated by the causality-aware approach will still be confounded. For instance, for the anticausal prediction tasks in model (\ref{eq:anticausal.task.biased}) the unmeasured confounders of the feature/response relationship will still confound the predictions.
\begin{equation}
{\footnotesize
\xymatrix@-1.3pc{
(M_a) & *+[F-:<10pt>]{\bfC} \ar[dl] \ar[drrr] \ar[ddr] &  & \bfU \ar[dlll] \ar[dr] \ar[ddl] & & (M_d) & *+[F-:<10pt>]{\bfC} \ar[drrr] \ar[ddr] &  & \bfU \ar[dlll] \ar[dr] \ar[ddl] & & (M_i) & *+[F-:<10pt>]{\bfC} \ar[drrr] &  & \bfU \ar[dlll] \ar[dr] \ar[ddl] & \\
*+[F-:<10pt>]{\bfX} & & & & *+[F-:<10pt>]{Y} \ar[llll] \ar[dll] & *+[F-:<10pt>]{\bfX^\ast} & & & & *+[F-:<10pt>]{Y} \ar[llll] \ar[dll] & *+[F-:<10pt>]{\bfX^\ast} & & & & *+[F-:<10pt>]{Y} \ar[dll]  \\
& & *+[F-:<10pt>]{\bfM} \ar[ull] & & & & & *+[F-:<10pt>]{\bfM} & & & & & *+[F-:<10pt>]{\bfM^\ast} \ar[ull] & & \\
}}
\label{eq:anticausal.task.biased}
\end{equation}

Finally, observe that while so far we have discussed confounding of the feature/response relationship, it is also possible that the causal relations between features and mediators or between mediators and response are also influenced by confounders. If these confounders are unobserved, then we cannot identify the causal effects $\bfGamma_{XM}$ and $\bfGamma_{MY}$. Clearly, in the presence of unobserved confounding the causality-aware predictions will be biased, whenever the causal effects of interest are not identifiable.

\section{Proof of Theorem 1}

Before we present the proof, we first clarify that, in the multivariate case, the covariance between two vectors of random variables, $\bfA = (A_1, \ldots, A_{N_A})^T$ and $\bfB = (B_1, \ldots, B_{N_B})^T$, is given by the cross-covariance operator, $Cov(\bfA, \bfB)$, defined and the $N_A \times N_B$ matrix with elements $Cov(A_i, B_j)$.

For the proof we will use the following properties of the cross-covariance operator:
\begin{enumerate}
\item $Cov(\bfZ_1 + \bfZ_2, \bfZ_3) = Cov(\bfZ_1, \bfZ_3) + Cov(\bfZ_2, \bfZ_3)$,
\item $Cov(\bfA \, \bfZ_1, \bfB \, \bfZ_2) = \bfA \, Cov(\bfZ_1, \bfZ_2) \, \bfB^T$, where $\bfA$ and $\bfB$ are constant matrices
\item $Cov(\bfZ, \bfZ) = Cov(\bfZ)$, where $Cov(\bfZ)$ is the variance covariance matrix of $\bfZ$.
\end{enumerate}

The proof is straight forward, and follow directly from the above three properties. For completeness we restate the Theorem.


\begin{theorem}
Consider an anticausal prediction task:

\textit{(i) When the interest focus on the causal effects generated by the paths in $Y \rightarrow \bfX$. If $\bfX^\ast$ is given by $\bfX^\ast = \bfGamma_{XY} \, Y + \bfW_X$, then $Cov(\bfX^\ast, Y) = \bfGamma_{XY}$.}

\textit{(ii) When the interest focus on the causal effects generated by the paths in $Y \rightarrow \bfM \rightarrow \bfX$. If  $\bfX^\ast$ is given by $\bfX^\ast = \bfGamma_{XM} \, \bfM^\ast + \bfW_X$, and $\bfM^\ast = \bfGamma_{MY} \, Y + \bfW_M$, then $Cov(\bfX^\ast, Y) = \bfGamma_{XM} \, \bfGamma_{MY}$.}

\textit{(iii) When the interest focus on the spurious associations generated by the paths in $\bfX \leftarrow \bfC \rightarrow Y$. If  $\bfX^\ast$ is given by $\bfX^\ast = \bfGamma_{XC} \, \bfC + \bfW_X$, then $Cov(\bfX^\ast, Y) = \bfGamma_{XC} \, Cov(\bfC) \, \bfGamma_{YC}^T$.}
\end{theorem}


\begin{proof}
$ $

Result \textit{i}: If $\bfX^\ast = \bfGamma_{XY} \, Y + \bfW_X$, then,
\begin{align*}
Cov(\bfX^\ast, Y) &= Cov(\bfGamma_{XY} \, Y + \bfW_X, Y) \\
&= \bfGamma_{XY} \, Cov(Y, Y) \\
&= \bfGamma_{XY}
\end{align*}

Result \textit{ii}: If $\bfX^\ast = \bfGamma_{XM} \, \bfM^\ast + \bfW_X$ and $\bfM^\ast = \bfGamma_{MY} \, Y + \bfW_M$, then,
\begin{align*}
Cov(\bfX^\ast, Y) &= Cov(\bfGamma_{XM} \, \bfM^\ast + \bfW_X, Y) \\
&= \bfGamma_{XM} \, Cov(M^\ast, Y)  \\
&= \bfGamma_{XM} \, Cov(\bfGamma_{MY} \, Y + \bfW_M, Y) \\
&= \bfGamma_{XM} \, \bfGamma_{MY} \, Var(Y) \\
&= \bfGamma_{XM} \, \bfGamma_{MY}
\end{align*}

Result \textit{iii}: If $\bfX^\ast = \bfGamma_{XC} \, \bfC + \bfW_X$, then,
\begin{align*}
Cov(\bfX^\ast, Y) &= Cov(\bfGamma_{XC} \, \bfC + \bfW_X, Y) \\
&= \bfGamma_{XC} \, Cov(\bfC, Y)  \\
&= \bfGamma_{XC} \, Cov(\bfC, \bfGamma_{YC} \, \bfC + W_Y) \\
&= \bfGamma_{XC} \, Cov(\bfC, \bfC) \, \bfGamma_{YC}^T \\
&= \bfGamma_{XC} \, Cov(\bfC) \, \bfGamma_{YC}^T
\end{align*}
\end{proof}

\section{Expected MSE for arbitrary anticausal prediction tasks based on linear models}

Consider the arbitrary anticausal prediction task model in Figure \ref{fig:confounded.anticausal.example},
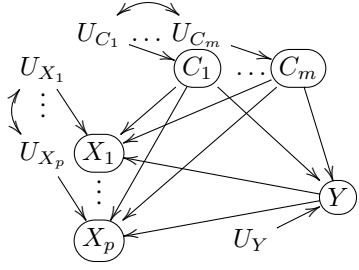
\begin{figure}[!h]
$$
\xymatrix@-2.2pc{
& U_{C_1} \ar[drr] \ar@/^1pc/@{<->}[rr] & \ldots & U_{C_m} \ar[drr] &&& \\
U_{X_1} \ar[ddr] \ar@/_1pc/@{<->}[dd] & & & *+[F-:<10pt>]{C_1} \ar[ddll] \ar[ddddll] \ar[dddrrr] & \ldots & *+[F-:<10pt>]{C_m} \ar[ddllll] \ar[ddddllll] \ar[dddr] \\
\vdots &&&&&&  \\
U_{X_p} \ar[ddr] & *+[F-:<10pt>]{X_1} &&&& \\
& \vdots & & & & & *+[F-:<10pt>]{Y} \ar[ulllll] \ar[dlllll] \\
& *+[F-:<10pt>]{X_p} &&& U_Y \ar[urr] && \\
}
$$
\vskip -0.1in
  \caption{Confounded anticausal prediction task example.}
  \label{fig:confounded.anticausal.example}
  \vskip -0.1in
\end{figure}
where the double arrows connecting the variables $\{U_{X_1}, \ldots, U_{X_p}\}$ (and $\{U_{C_1}, \ldots, U_{C_m}\}$) represent the fact that these error terms are correlated\footnote{Note that the above model might represent a reparameterization of a model with uncorrelated error terms and unknown causal relations among the $\bfX$ input variables, as well as, among the $\bfC$ confounder variables. As described in detail in the main text, for linear structural equation models, we can always reparameterize the original model in a way where the covariance structure among the input variables, as well as, the covariance structure among the confounder variables is pushed towards the respective error terms as illustrated in Figure \ref{fig:confounded.anticausal.example}.}. Without loss of generality assume that the data has been centered, so that the linear structural causal models describing the data generation process are given by,
\begin{align}
C_j &= U_{C_j}~, \\
Y &= \sum_i \beta_{Y{C_i}} \, C_i + U_Y~, \\
X_j &= \beta_{{X_j}Y} Y + \sum_i \beta_{{X_j}{C_i}} \, C_i + U_{X_j}~, \label{eq:feature.model}
\end{align}
for $j = 1, \ldots, p$ and $i = 1, \ldots, m$. The causality-aware features are estimated as,
\begin{align}
\hat{X}_j^\ast &= X_j - \sum_i \hat{\beta}_{{X_j}{C_i}} C_i~,
\end{align}
and converge asymptotically to,
\begin{align}
X_j^\ast \, &= \, X_j - \sum_i \beta_{{X_j}{C_i}} \, C_i \, = \, \beta_{{X_j}Y} \, Y + U_{X_j}~. \label{eq:feature.model.ca}
\end{align}




Now, let $\hat{Y} = \bfX_{ts} \hat{\bfBeta}^{tr}$ represent the prediction of a linear regression model, where $\bfX_{ts}$ represents the test set features, and $\hat{\bfBeta}^{tr}$ represents the regression coefficients estimated from the training set. By definition, the expected mean squared error of the prediction is given by,
\begin{align*}
E[MSE] &= E[(Y_{ts} - \hat{Y})^2] = E[Y^2_{ts}] + E[\hat{Y}^2] - 2 E[\hat{Y} Y_{ts}] \\
&= Var(Y_{ts}) + E[\hat{Y}^2] - 2 Cov(\hat{Y}, Y_{ts})~,
\end{align*}
since $E[Y_{ts}] = 0$. Direct computation shows that,
\begin{align*}
E[\hat{Y}^2] &= E[(\sum_{j=1}^{p} X_{j,ts} \hat{\beta}_j^{tr})^2] = \sum_{j=1}^{p} (\hat{\beta}_j^{tr})^2 Var(X_{j,ts}) \, + 2 \sum_{j < k} \hat{\beta}_j^{tr} \hat{\beta}_k^{tr} Cov(X_{j,ts}, X_{k,ts})~,
\end{align*}
and,
\begin{align*}
Cov(\hat{Y}, Y_{ts}) &= \sum_{j=1}^{p} \hat{\beta}_j^{tr} Cov(X_{j,ts}, Y_{ts})~,
\end{align*}
so that,
\begin{align*}
E[MSE] \, &=  \, Var(Y_{ts}) \, + \sum_{j=1}^{p} (\hat{\beta}_j^{tr})^2 Var(X_{j,ts}) \, + \\
&+ 2 \sum_{j < k} \hat{\beta}_j^{tr} \hat{\beta}_k^{tr} Cov(X_{j,ts}, X_{k,ts}) - 2 \sum_{j=1}^{p} \hat{\beta}_j^{tr} Cov(X_{j,ts}, Y_{ts})~.
\end{align*}

Next, we derive the expressions for $Var(X_{j,ts})$, $Cov(X_{j,ts}, X_{k,ts})$, and $Cov(X_{j,ts}, Y_{ts})$ and show that they still depend on $Cov(Y_{ts}, C_{i,ts})$. From equation (\ref{eq:feature.model}) we have that,
\begin{align*}
Var(X_{j,ts}) &= Var(\beta_{{X_j}Y} \, Y_{ts} + \sum_i \beta_{{X_j}{C_i}} \, C_{i,ts} + U_{X_j}^{ts}) \\
&= \sigma^2_{X_j} + \beta_{{X_j}Y}^2 \, Var(Y_{ts}) + \sum_i \beta_{{X_j}{C_i}}^2 \, Var(C_{i,ts}) \, + \\
&\;\;\;\; + 2 \, \sum_{i < i'} \beta_{{X_j}{C_i}} \, \beta_{{X_j}{C_{i'}}} \, Cov(C_{i,ts}, C_{i',ts}) + \, 2 \, \beta_{{X_j}Y} \sum_i \beta_{{X_j}{C_i}} \, Cov(Y_{ts}, C_{i,ts})~, \\
\end{align*}
\begin{align*}
Cov(X_{j,ts}, X_{k,ts}) &= Cov(\beta_{{X_j}Y} \, Y_{ts} + \sum_i \beta_{{X_j}{C_i}} \, C_{i,ts} + U_{X_j}^{ts} \; , \; \beta_{{X_k}Y} \, Y_{ts} + \sum_i \beta_{{X_k}{C_i}} \, C_{i,ts} + U_{X_k}^{ts}) \\
&= \beta_{{X_j}Y} \, \beta_{{X_k}Y} \, Var(Y_{ts}) + \beta_{{X_j}Y} \, \sum_i \beta_{{X_k}{C_i}} \, Cov(Y_{ts}, C_{i,ts}) + \\
&\;\;\;\;+ \beta_{{X_k}Y} \, \sum_i \beta_{{X_j}{C_i}} \, Cov(Y_{ts}, C_{i,ts}) + \sum_i \sum_{i'} \beta_{{X_j}{C_i}} \, \beta_{{X_k}{C_{i'}}} \, Cov(C_{i,ts}, C_{i',ts}) + \\
&\;\;\;\;+ Cov(U_{X_j}^{ts}, U_{X_k}^{ts})
\end{align*}

\begin{align*}
Cov(X_{j,ts}, Y_{ts}) &= Cov(\beta_{{X_j}Y} \, Y_{ts} + \sum_i \beta_{{X_j}{C_i}} \, C_{i,ts} + U_{X_j}, Y_{ts}) \\
&= \beta_{{X_j}Y} \, Var(Y_{ts}) + \sum_i \beta_{{X_j}{C_i}} \, Cov(Y_{ts}, C_{i,ts})~,
\end{align*}
showing that these three quantities still depend on $Cov(Y_{ts}, C_{i,ts})$ (in addition to depending on $Var(Y_{ts})$, $Var(C_{ts})$, and $Cov(C_{i,ts}, C_{i',ts})$). This observation implies that the $E[MSE]$ will still be unstable w.r.t. shifts in these quantities, even when the regression model is trained in unconfounded data (a situation where the estimates $\hat{\beta}_j^{tr}$ are not influenced by spurious associations generated by the confounders). This explains why it is not enough to deconfound the training features alone. While training a regression model using deconfounded features allows us to estimate deconfounded model weights\footnote{Note that the weights $\hat{\bfBeta}_{tr}$ are not causal effects, since they represent the coefficients of the regression of $Y_{tr}$ on $\bfX_{tr}$, while in the true data generation process $Y_{tr}$ is the independent variable and $\bfX_{tr}$ represents the dependent variables. Still, the estimate $\hat{\bfBeta}_{tr}$ will not absorb spurious associations when the model is trained with unconfounded data.},  $\hat{\bfBeta}^{tr}$, the prediction $\hat{Y} = \bfX_{ts} \hat{\bfBeta}^{tr}$ is a function of both the trained model $\hat{\bfBeta}^{tr}$ and the test set features, $\bfX_{ts}$. As a consequence, if we do not deconfound the test set features, the expected MSE will still be influenced by the confounders (since, in anticausal prediction tasks, the features are functions of both the confounder and outcome variables).


The expected MSE of models trained with test set features processed according to the causality-aware approach, on the other hand, do not depend on $Cov(Y_{ts}, C_{i,ts})$, $Var(C_{ts})$, or $Cov(C_{i,ts}, C_{i',ts})$, since the approach also deconfounds the test set features. Note that direct computation of $Var(X_{j,ts}^\ast)$, $Cov(X_{j,ts}^\ast, X_{k,ts}^\ast)$, and $Cov(X_{j,ts}^\ast, Y_{ts})$ based on the causality-aware features, $X_{j,ts}^\ast = \beta_{{X_j}Y} \, Y_{ts} + U_{X_j}^{ts}$, shows that,
\begin{align*}
Var(X_{j,ts}^\ast) &= Var(\beta_{{X_j}Y} \, Y_{ts} + U_{X_j}^{ts}) = \sigma^2_{X_j} + \beta_{{X_j}Y}^2 \, Var(Y_{ts})~,
\end{align*}
\begin{align*}
Cov(X_{j,ts}^\ast, X_{k,ts}^\ast) &= Cov(\beta_{{X_j}Y} \, Y_{ts} + U_{X_j}^{ts}, \beta_{{X_k}Y} \, Y_{ts} + U_{X_k}^{ts}) \\
&= \beta_{{X_j}Y} \, \beta_{{X_k}Y} \, Var(Y_{ts}) + Cov(U_{X_j}^{ts}, U_{X_k}^{ts})~,
\end{align*}
\begin{align*}
Cov(X_{j,ts}^\ast, Y_{ts}) &= Cov(\beta_{{X_j}Y} \, Y_{ts} + U_{X_j}^{ts}, Y_{ts}) = \beta_{{X_j}Y} \, Var(Y_{ts})~,
\end{align*}
no longer depend on $Cov(Y_{ts}, C_{i,ts})$, $Var(C_{ts})$, or $Cov(C_{i,ts}, C_{i',ts})$, so that the approach will be stable against shifts in these quantities. Observe, nonetheless, that it will still be influenced by shifts in $Var(Y_{ts})$. (We point out, however, that the dependence of $E[MSE]$ on $Var(Y_{ts})$ is, in general, unavoidable since, by definition, $E[MSE]= Var(Y_{ts}) + E[\hat{Y}^2] - 2 Cov(\hat{Y}, Y_{ts})$.)

\section{Extensions to arbitrary performance metrics and arbitrary structural causal models}

Here, we extend the argument presented in the previous section to arbitrary performance metrics and arbitrary structural causal models.

Let $M = h_1(Y_{ts}, \hat{Y})$ represent an arbitrary performance metric, and let $\hat{Y} = h_2(\bfomega_{tr}, \bfX_{ts}) = h_2(\bfomega_{tr}, X_{1,ts}, \ldots, X_{p,ts})$ represent a prediction generated with an arbitrary ML model $\bfomega_{tr}$. Note that $\bfomega_{tr} = h_3(\bfX_{tr}, Y_{tr})$ is a function of the training data. Assume the features $X_{j}$ are generated by an arbitrary structural causal model $X_j = f(Y, \bfC, U_{X_j})$. Then the expected value of $M$, with respect to the test set data distribution is given by,
\begin{align*}
E[M] &= E[h_1(Y_{ts}, \hat{Y})] \\
&= E[h_1(Y_{ts}, h_2(\bfomega_{tr}, X_{1,ts}, \ldots, X_{p,ts}))] \\
&= E[h_1(Y_{ts}, h_2(\bfomega_{tr}, f(Y_{ts}, \bfC_{ts}, U_{X_1}^{ts}), \ldots, f(Y_{ts}, \bfC_{ts}, U_{X_p}^{ts})))]~,
\end{align*}
showing that even when we train the model $\bfomega_{tr}$ using deconfounded training data, we have that $E[M]$ is still a function of $\bfC_{ts}$, and will be unstable with respect to shifts in $P(\bfC_{ts}, Y_{ts})$.


Observe, however, that if we are able to deconfound the test set features, so that the counterfactual features $X_{j,ts}^\ast = f^\ast(Y_{ts}, U_{X_j}^{ts})$ are no longer a function of the confounders, then we have that,
\begin{align*}
E[M] &= E[h_1(Y_{ts}, \hat{Y}^\ast)] \\
&= E[h_1(Y_{ts}, h_2(\bfomega_{tr}, X_{1,ts}^\ast, \ldots, X_{p,ts}^\ast))] \\
&= E[h_1(Y_{ts}, h_2(\bfomega_{tr}, f^\ast(Y_{ts}, U_{X_1}^{ts}), \ldots, f^\ast(Y_{ts}, U_{X_p}^{ts})))]~,
\end{align*}
will no longer depend on $\bfC_{ts}$. Note that while the predictive performance will still depend on the distribution of $Y_{ts}$ and, therefore, will still be unstable with respect to shifts in the marginal distribution $P(Y_{ts})$, the approach will still be stable with respect to shifts in the conditional distribution $P(\bfC_{ts} \mid Y_{ts})$.

\section{Additional details - synthetic data experiments}

In our experiments, we compare the causality-aware approach against two ``archetypical" baselines: (1) one representing adjustment approaches that remove the causal effect of the confounders from the features, denoted \textit{baseline 1}; and (2) another representing approaches that remove the association between the confounders and the output, denoted \textit{baseline 2}. In both baseline approaches we adjust the training data but not the test set. Note that for both of these baselines, while the training data is unconfounded, the test data is still confounded. For the causality-aware approach, on the other hand, we generated confounded training and test sets and then apply our adjustment for both the training and test sets.


The confounded data is generated from the model,
$$
\xymatrix@-1.2pc{
& & & *+[F-:<10pt>]{C} \ar[dll]_{\beta_{{X_1}C}} \ar[ddll] \ar[ddddll] \ar[ddr]^{\beta_{YC}} & U_{C} \ar[l] \ar@/^1pc/@{<->}[ddr] & \\
U_{X_1} \ar[r] \ar@/_1pc/@{<->}[dd] \ar@/_1.5pc/@{<->}[ddd] & *+[F-:<10pt>]{X_{1}} &&&& \\
& *+[F-:<10pt>]{X_{2}} & & & *+[F-:<10pt>]{Y} \ar[ulll] \ar[lll] \ar[ddlll]^{\beta_{{X_{10}}Y}} & U_Y \ar[l] \\
U_{X_{2}} \ar[ru] & \vdots &&&& \\
U_{X_{10}} \ar[r] & *+[F-:<10pt>]{X_{10}} &&&  \\
}
$$
where we change the covariance of the error terms $U_C$ and $U_Y$ in order to simulate the effects of selection biases in the joint distribution $P(C, Y)$.

The model is described by the following set of linear structural causal equations,
\begin{align}
C &= U_{C}~, \label{eq:C.regr.model} \\
Y &= \beta_{YC} \, C + U_Y~, \label{eq:Y.regr.model} \\
X_{j} &= \beta_{{X_j}Y} \, Y + \beta_{{X_j}{C}} \, C + U_{X_j}~, \label{eq:X.regr.model}
\end{align}
for $j = 1, \ldots, 10$, and where the error terms $U_C$ and $U_X$ are distributed according to,
\begin{equation}
\begin{pmatrix}
U_{C} \\
U_{Y} \\
\end{pmatrix} \,
\sim \, \mbox{N}_2\left(
\begin{pmatrix}
0 \\
0 \\
\end{pmatrix}\, , \,
\begin{pmatrix}
\phi_{CC} & \phi_{CY} \\
\phi_{CY} & \phi_{YY} \\
\end{pmatrix} \right)~, \label{eq:errors.C.Y.distr}
\end{equation}
and $\bfU_X = (U_{X_1}, \ldots, U_{X_{10}})^T$ is distributed according to a multivariate normal distribution,
\begin{equation}
\bfU_X \, \sim \, N_{10}(\bfzero \, , \ \bfSigma_{\bfU_X})~, \label{eq:errors.X.distr}
\end{equation}
where the $ij$th entry of the covariance matrix $\bfSigma_{\bfU_X}$ is given by 1 for $i = j$, and by $\rho^{|i-j|}$ for $i \not= j$.

Note that, for the above model, we have that,
\begin{align}
Var(C) &= \phi_{CC}~, \\
Cov(Y, C) &= Cov(\beta_{YC} \, C + U_Y, C) = \beta_{YC} \, Var(C) + Cov(U_Y, C) \nonumber \\
&= \beta_{YC} \, \phi_{CC} + \phi_{CY}~, \\
Var(Y) &= Var(\beta_{YC} \, C + U_Y) = \beta_{YC}^2 \, Var(C) + Var(U_Y) + 2 \, \beta_{YC} \, Cov(C, U_Y) \nonumber \\
&= \beta_{YC}^2 \, \phi_{CC} + \phi_{YY} + 2 \, \beta_{YC} \, \phi_{CY}~,
\end{align}
so that for fixed values of $Var(C)$, $Cov(Y, C)$, $Var(Y)$, and $\beta_{{Y}{C}}$ we can determine the values of $\phi_{CC}$, $\phi_{CY}$, and $\phi_{YY}$ as follows,
\begin{align}
\phi_{CC} &= Var(C)~, \label{eq:phi.CC} \\
\phi_{CY} &= Cov(Y, C) - \beta_{YC} \, Var(C)~, \label{eq:phi.CY} \\
\phi_{YY} &= Var(Y) - \beta_{YC}^2 \, Var(C) -2 \, \beta_{YC} \, Cov(Y, C)~. \label{eq:phi.YY}
\end{align}



In our experiments, we simulate training and test set data as follows:
\begin{enumerate}
\item Sample the simulation parameters $\beta_{{X_j}Y}$, $\beta_{{X_j}{C}}$, and $\beta_{{Y}{C}}$ from a $U(-1, 1)$ distribution, and $\rho$ from a $U(-0.5, 0.5)$ distribution.
\item Given the fixed values for $Var(C_{tr})$, $Cov(Y_{tr}, C_{tr})$, and $Var(Y_{tr})$, and the sampled value for $\beta_{{Y}{C}}$, we compute $\phi_{CC}$, $\phi_{CY}$, and $\phi_{YY}$ as described in equations (\ref{eq:phi.CC}), (\ref{eq:phi.CY}), and (\ref{eq:phi.YY}).
\item Sample the error terms $U_C^{tr}$ and $U_Y^{tr}$ according to (\ref{eq:errors.C.Y.distr}), and the error terms $\bfU_X^{tr}$ according to (\ref{eq:errors.X.distr}).
\item Simulate 3 separate training sets, the confounded one (where we apply the causality-aware adjustment), and the baseline 1 and baseline 2 training sets (using the exact same error terms sampled in the previous step). The confounded training set was generated according to the following model,
    \begin{align}
    C_{tr} &= U_{C}^{tr}~, \\
    Y_{tr} &= \beta_{YC} C_{tr} + U_Y^{tr}~, \\
    X_{j,tr} &= \beta_{{X_j}Y} Y_{tr} + \beta_{{X_j}{C}} C_{tr} + U_{X_j}^{tr}~.
    \end{align}
    The baseline 1 training data was generated according to the model,
    \begin{align}
    C_{tr} &= U_{C}^{tr}~, \\
    Y_{tr} &= \beta_{YC} C_{tr} + U_Y^{tr}~, \\
    X_{j,tr} &= \beta_{{X_j}Y} Y_{tr} + U_{X_j}^{tr}~,
    \end{align}
    while the baseline 2 training data was generated according to the model,
    \begin{align}
    C_{tr} &= U_{C}^{tr}~, \\
    Y_{tr} &= U_Y^{tr}~, \\
    X_{j,tr} &= \beta_{{X_j}Y} Y_{tr} + \beta_{{X_j}{C}} C_{tr} + U_{X_j}^{tr}~.
    \end{align}
\item Simulate 9 distinct confounded test sets (indexed by ${ts_k}$, for $k = 1, \ldots, 9$). Each test set is simulated as follows:
\begin{enumerate}
\item Given the fixed values for $Var(C_{{ts}_k})$, $Cov(Y_{{ts}_k}, C_{{ts}_k})$, and $Var(Y_{{ts}_k})$, and the sampled value for $\beta_{{Y}{C}}$, we compute $\phi_{CC}$, $\phi_{CY}$, and $\phi_{YY}$ as described in equations (\ref{eq:phi.CC}), (\ref{eq:phi.CY}), and (\ref{eq:phi.YY}).
\item Sample the error terms $U_C^{{ts}_k}$ and $U_Y^{{ts}_k}$ according to (\ref{eq:errors.C.Y.distr}), and the error terms $\bfU_X^{{ts}_k}$ according to (\ref{eq:errors.X.distr}).
\item Simulate the test set data according to the model,
    \begin{align}
    C_{{ts}_k} &= U_{C}^{{ts}_k}~, \\
    Y_{{ts}_k} &= \beta_{YC} C_{{ts}_k} + U_Y^{{ts}_k}~, \\
    X_{j,{ts}_k} &= \beta_{{X_j}Y} Y_{{ts_k}} + \beta_{{X_j}{C}} C_{{ts_k}} + U_{X_j}^{{ts}_k}~,
    \end{align}
\end{enumerate}
\end{enumerate}
Note that, in order to generate dataset shifts in $P(C, Y)$, we allow $Var(C)$, $Cov(Y, C)$, and $Var(Y)$ to vary between the training and test sets. However, in order to maintain the stability of $P(\bfX \mid C, Y)$ we use the same sampled values of $\beta_{{X_j}Y}$, $\beta_{{X_j}{C}}$, $\beta_{{Y}{C}}$ and $\rho$ in the generation of the training and test sets.



In order to illustrate the influence of $Var(Y_{ts})$ in the stability of the predictions, we performed two experiments. In the first, we kept the $Var(Y_{ts})$ constant across the test sets. In the second, we increased $Var(Y_{ts})$ across the test sets. Each of these experiments were based on 1000 simulations. For each simulation replication we:
\begin{enumerate}
\item Generate the 3 training sets ($n = 1,000$) by setting $Var(C_{tr}) = 1$, $Cov(Y_{tr}, C_{tr}) = 0.8$, and $Var(Y_{tr}) = 1$ and then simulating the data as described above.
\item Generate 9 distinct test sets (each containing $n = 1,000$ samples). Each test set was generated with an increasing amount of shift in the $P(C, Y)$ distribution. In the first experiment this was accomplished by varying $Cov(Y_{{ts}_k}, C_{{ts}_k})$ according to $\{0.8$, 0.6, 0.4, 0.2, 0.0, -0.2, -0.4, -0.6, $-0.8\}$ across the 9 test sets, and by varying $Var(C_{{ts}_k})$ according to $\{1.00$, 1.25, 1.50, 1.75, 2.00, 2.25, 2.50, 2.75, $3.00\}$, while keeping $Var(Y_{{ts}_k})$ fixed at 1 for all $k$. In the second experiment, we varied $Cov(Y_{{ts}_k}, C_{{ts}_k})$ as before, but kept $Var(C_{{ts}_k})$ fixed at 1, while varying $Var(Y_{{ts}_k})$ according to $\{1.00$, 1.25, 1.50, 1.75, 2.00, 2.25, 2.50, 2.75, $3.00\}$ across the test sets.
\item For the causality-aware approach we: adjust the confounded training set, and each of the 9 confounded test sets; fit a regression model on the adjusted training set data; use the same trained model to predict on the 9 adjusted test sets; and evaluate the test set performances using MSE.
\item For the baseline 1 approach we: fit a regression model to the (unconfounded) baseline 1 training set; use the trained model to predict on the 9 confounded test sets; and evaluate the test set performances using MSE.
\item For the baseline 2 approach we: fit a regression model to the (unconfounded) baseline 2 training set; use the trained model to predict on the 9 confounded test sets; and evaluate the test set performances using MSE.
\item For the ``no adjustment" approach we: fit a regression model to the confounded training data; use the trained model to predict on the 9 confounded test sets; and evaluate the test set performances using MSE.
\end{enumerate}
Note that the first test set is generated using the same values of $Cov(Y, C)$, $Var(C)$, and $Var(Y)$ as the training set, so that it illustrates the case where the training and test sets are independent and identically distributed. (Observe that in this setting, performing confounding adjustment may decrease the predictive performance of the learner in situations where the confounder strengths the association between the features and the outcome variable.)

Next, we present a few important remarks.

\begin{enumerate}
\item Note that the (``archetypical") baseline 1 approach is meant to represent methods that attempt to remove the causal effects of the confounders from the features in the training set alone. This includes a poor man's version of the causality-aware approach where we do not process the test set features. Our goal here is to illustrate that while it might seen intuitive that training a learner on unconfounded data will prevent it from learning the confounding signal and, therefore, will lead to more stable predictions in shifted target populations, the unconfounded trained model, $\hat{\bfBeta}^{tr}$ is only one component of the prediction, $\hat{Y} = \bfX_{ts} \hat{\bfBeta}^{tr}$, so that better stability can be achieved by deconfounding the test set features, $\bfX_{ts}$, as well.

\item Second, note that the baseline 2 approach is meant to represent methods that attempt to remove the association between the confounders and the outcome. Those include approaches such as propensity scores for continuous variables~\cite{hirano2004}, covariate balancing propensity score methods for continuous variables~\cite{fong2018}, or standard propensity score matching applied to dichotomized outcome data\footnote{For classification tasks these methods include standard matching and IPW by propensity score methods.}. As described before, rather than implementing these methods, we simulate unconfounded training data where the output is statistically independent from the confounders, which mimics the case where these adjustments worked perfectly. (Observe that, in the particular context of classification tasks, removing the association between labels and confounders represents a common strategy to combat discrimination in fairness research, where data pre-processing techniques such as re-weighting and (under-) over-sampling are applied to the training data alone, in order remove the association between sensitive variables (i.e., confounders) and the classifier labels~\cite{calders2009,kamiran2012})

\item Third, it is important to point out that several approaches proposed in the stable prediction literature are not applicable in our illustrations. For instance, in the context of classification tasks, approaches such as invariant risk minimization~\cite{irm2019} or invariant causal prediction~\cite{peters2016} rely on training data from multiple training sets while our approach focus on a single training set. Furthermore, stable prediction approaches~\cite{kuang2018,kuang2020}, which only require a single training set, can only be applied in causal prediction tasks, while our illustrations focus on anticausal tasks.

\item Fourth, observe that our approach assumes that $P(\bfX \mid \bfC, Y)$ is stable across the test set domains. This assumption is reasonable in several application domains. For instance, in health diagnostic applications, where the goal is to classify (for example) mild vs severe cases of a given disease, using the disease symptoms as inputs, we have that $P(\bfX \mid \bfC, Y)$ tends to be stable for demographic confounders such as age and gender. Note that this distribution would be unstable in the less likely scenario where the individuals in the training set have different symptom severities (caused by age, gender and disease status) than individuals in distinct test sets, pointing to biological/physiological differences between the individuals in training and testing populations. Dataset shifts on $P(\bfC, Y)$, on the other hand, are much more commonly observed in health applications, because selection biases during data collection often mean that the $P(\bfC, Y)$ distribution in the target/test populations are shifted relative to the training data.
    
\item Finally, note that application of counterfactual normalization~\cite{subbaswamy2018} approach to the particular causal graph used in our experiments would augment the causal graph with the counterfactual variables $X_{j}(C = \emptyset)$ (representing the values of $X_j$ we would have seen, had $C$ not being a parent of $X_j$), and would return the counterfactual variables $X_{j}(C = \emptyset)$ as the stable set for predicting $Y$. Because (for this particular example) the counterfactual features, $X_{j}(C = \emptyset)$, are computed in exactly the same way as the causality-aware training features, $X_{j}^\ast = X_{j} - \hat{\beta}_{{X_j}C} \, C$, it follows that application of counterfactual normalization would produce the same results as the causality-aware approach.
\end{enumerate}


\section{The causal prediction task case}

In this paper, we have focused in anticausal prediction tasks. A few analogous results are, nonetheless,
\begin{figure}[!h]
$$
{\footnotesize
\xymatrix@-1.4pc{
& *+[F-:<10pt>]{\bfC} \ar[dl] \ar[dr] \ar[dd] & \\
*+[F-:<10pt>]{\bfX} \ar[rr] \ar[dr] & & *+[F-:<10pt>]{Y} \\
& *+[F-:<10pt>]{\bfM} \ar[ur] & \\
}}
$$
\vskip -0.1in
  \caption{Causal prediction task.}
  \label{fig:causal.pred.task}
\end{figure}
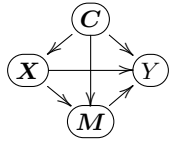
available for causal prediction tasks (i.e., prediction tasks where the inputs influence the outcome). In the next subsections, we present these results.

\subsection{Reparameterization in causal prediction tasks}

For the causal prediction task presented in Figure \ref{fig:causal.pred.task} we have that the joint distribution factorizes as,
\begin{equation*}
P(\bfC) P(\bfX \mid \bfC) P(\bfM \mid \bfC, \bfX) P(Y \mid \bfC, \bfM, \bfX)~,
\end{equation*}
where each component is described by the structural model,
\begin{align*}
\bfC &= \bfTheta_{CC} \, \bfC + \bfU_C~, \\
\bfX &= \bfTheta_{XX} \, \bfX + \bfTheta_{XC} \, \bfC + \bfU_X~, \\
\bfM &= \bfTheta_{MM} \, \bfM + \bfTheta_{MC} \, \bfC + \bfTheta_{MX} \, \bfX + \bfU_M~, \\
Y &= \bfTheta_{YC} \, \bfC + \bfTheta_{YM} \, \bfM + \bfTheta_{YX} \, \bfX + U_Y~,
\end{align*}
which can also be reparameterized as,
\begin{align*}
\bfC &= \bfW_C~, \\
\bfX &= \bfGamma_{XC} \, \bfC + \bfW_X~, \\
\bfM &= \bfGamma_{MC} \, \bfC + \bfGamma_{MX} \, \bfX + \bfW_M~, \\
Y &= \bfGamma_{YC} \, \bfC + \bfGamma_{YM} \, \bfM + \bfGamma_{YX} \, \bfX + W_Y~,
\end{align*}
where $\bfGamma_{MX} = (\bfI - \bfTheta_{MM})^{-1} \bfTheta_{MX}$, $\bfGamma_{YC} = \bfTheta_{YC}$, $\bfGamma_{YM} = \bfTheta_{YM}$, $\bfGamma_{YX} = \bfTheta_{YX}$, $W_Y = U_Y$, and the other parameters and error terms are given as before.

\section{Estimation of the causal effects in the causal task, and remarks on identification issues}

For the causal prediction task, we regress the response on the confounders, mediators, and features,
\begin{equation*}
Y = \sum_{k=1}^{n_C} \gamma_{_{{Y}{C_k}}} \, C_k + \sum_{k=1}^{n_M} \gamma_{_{{Y}{M_k}}} \, M_k + \sum_{k=1}^{n_X} \gamma_{_{{Y}{X_k}}} \, X_k + W_{Y}~,
\end{equation*}
and then generate the counterfactual response by adding back $\hat{W}_{Y}$ to a linear predictor containing only the causal effects of interest. In particular, we can generate counterfactual response data that captures the predictive performance due to direct causal effects, indirect causal effects, or to confounding, using, respectively,
\begin{align*}
\hat{Y}^\ast &= \hat{\bfGamma}_{YX} \bfX + \hat{W}_Y~, \\
\hat{Y}^\ast &= \hat{\bfGamma}_{YM} \, \hat{\bfM}^\ast + \hat{W}_Y~, \\
\hat{Y}^\ast &= \hat{\bfGamma}_{YC} \, \bfC + \hat{W}_Y~.
\end{align*}


Under the assumption that all the confounders and mediators are observed, we can identify the direct and indirect causal effects of the features on the response. In particular, a simple least squares estimation procedure provides consistent estimates of these causal effects\footnote{Here, we assume that the number of samples is larger than the number of covariates in the regression fits, and that multicolinearity is not an issue too.}. To see why, note that for the reparameterized model, if all confounders and mediators are observed, it follows from the Markov property of DAGs that $Y = f_{Y}(\bfC, \bfM, \bfX, W_{Y}) = f_{Y}(pa(Y), W_{Y})$. (Here, $f_Y$ represent a linear structural causal model). Hence, it follows that, when we regress $Y$ on the elements of $\bfC$, $\bfM$, and $\bfX$ only the coefficients associated with the parents of $Y$ will be statistically different from zero (for large enough sample sizes). Therefore, in practice, we don't need to know before hand which variables are the parents of $Y$. The parent set will be learned automatically from the data by the regression model fit.

Observe, as well, that even if the mediators are unobserved, but the confounders are still observed, we can still identify total causal effects. For instance, in causal tasks we have that,
\begin{align*}
Y &= \bfGamma_{YC} \, \bfC + \bfGamma_{YM} \, \bfM + \bfGamma_{YX} \, \bfX + W_Y~, \\
&= \bfGamma_{YC} \, \bfC + \bfGamma_{YM} \, (\bfGamma_{MC} \, \bfC + \bfGamma_{MX} \, \bfX + \bfW_M) + \bfGamma_{YX} \, \bfX + W_Y~, \\
&= \underbrace{(\bfGamma_{YC} + \bfGamma_{YM} \, \bfGamma_{MC})}_{\tilde{\bfGamma}_{YC}} \, \bfC + \underbrace{(\bfGamma_{YX} + \bfGamma_{YM} \, \bfGamma_{MX})}_{\tilde{\bfGamma}_{YX}} \, \bfX + \underbrace{\bfGamma_{YM} \, \bfW_M + W_Y}_{\tilde{W}_Y}~, \\
&= \tilde{\bfGamma}_{YC} \, \bfC + \tilde{\bfGamma}_{YX} \, \bfX + \tilde{W}_Y~,
\end{align*}
where $\tilde{\bfGamma}_{YX} = \bfGamma_{YX} + \bfGamma_{YM} \, \bfGamma_{MX}$ represents the total causal effect of $\bfX$ on $Y$, as represented in the DAG $M_b$ in in the causal task model (\ref{eq:causal.task.total}).
\begin{equation}
\xymatrix@-0.0pc{
(M_a) & *+[F-:<10pt>]{\bfC} \ar[dl]_{\bfGamma_{XC}} \ar[dr]^{\bfGamma_{YC}} \ar@/_1.5pc/[dd]|-{\bfGamma_{MC}} & & (M_b) & *+[F-:<10pt>]{\bfC} \ar[dl]_{\bfGamma_{XC}} \ar[dr]^{\tilde{\bfGamma}_{YC}} & \\
*+[F-:<10pt>]{\bfX} \ar@/_0.75pc/[rr]|-{\bfGamma_{YX}} \ar[dr]_{\bfGamma_{MX}} & & *+[F-:<10pt>]{Y}  & *+[F-:<10pt>]{\bfX} \ar[rr]_{\tilde{\bfGamma}_{YX}} && *+[F-:<10pt>]{Y} \\
& *+[F-:<10pt>]{\bfM} \ar[ur]_{\bfGamma_{YM}} & & & & \\
}
\label{eq:causal.task.total}
\end{equation}

On the other hand, if the mediators are observed, but some the confounders are unobserved, then neither the direct, the indirect, or the total causal effects are identifiable, and the predictions generated by the causality-aware approach will still be confounded. For instance, for the causal prediction tasks in model (\ref{eq:causal.task.biased}), we have that the unobserved confounders, $\bfU$, still confound the direct causal effect of $\bfX$ on $Y^\ast$ in model $M_d$, and the indirect causal effect in model $M_i$. As a consequence, the spurious associations contributed by $\bfU$ will still bias the predictions from models trained with the counterfactual data.
\begin{equation}
{\footnotesize
\xymatrix@-1.3pc{
(M_a) & *+[F-:<10pt>]{\bfC} \ar[dl] \ar[drrr] \ar[ddr] &  & \bfU \ar[dlll] \ar[dr] \ar[ddl] & & (M_d) & *+[F-:<10pt>]{\bfC} \ar[dl] \ar[ddr] &  & \bfU \ar[dlll] \ar[dr] \ar[ddl] & & (M_i) & *+[F-:<10pt>]{\bfC} \ar[dl] &  & \bfU \ar[dlll] \ar[dr] \ar[ddl] & \\
*+[F-:<10pt>]{\bfX} \ar[rrrr] \ar[drr] & & & & *+[F-:<10pt>]{Y} & *+[F-:<10pt>]{\bfX} \ar[rrrr] \ar[drr] & & & & *+[F-:<10pt>]{Y^\ast} & *+[F-:<10pt>]{\bfX} \ar[drr] & & & & *+[F-:<10pt>]{Y^\ast}  \\
& & *+[F-:<10pt>]{\bfM} \ar[urr] & & & & & *+[F-:<10pt>]{\bfM} & & & & & *+[F-:<10pt>]{\bfM^\ast} \ar[urr] & & \\
}}
\label{eq:causal.task.biased}
\end{equation}

Finally, observe that while so far we have discussed confounding of the feature/response relationship, it is also possible that the causal relations between features and mediators or between mediators and response are also influenced by confounders. If these confounders are unobserved, then we cannot identify the causal effects $\bfGamma_{MX}$ and $\bfGamma_{YM}$. Clearly, in the presence of unobserved confounding the causality-aware predictions will be biased, whenever the causal effects of interest are not identifiable.

\subsection{Causality-aware predictions in causal prediction tasks - the univariate case}

Consider a causal prediction task where the goal is to build a ML model whose predictive performance is only informed by the direct causal effect of $X$ on $Y$. We can simulate counterfactual response data, $Y^\ast$, according to the twin network in Figure \ref{fig:twins}a so that,
\begin{align}
Cov(X, Y^\ast) &= Cov(X, \, \theta_{YX} X + U_Y) = \theta_{YX} \, Var(X) = \theta_{YX}~, \label{eq:causal.direct.effect}
\end{align}

Now, consider a causal prediction task where the goal is to build a ML model whose predictive performance is only informed by the indirect causal effect of $X$ on $Y$. Now, we can simulate counterfactual response data, $Y^\ast$, according to the twin network in Figure \ref{fig:twins}b so that,
\begin{align}
Cov(X, Y^\ast) &= Cov(X, \, \theta_{YM} M^\ast + U_Y) = \theta_{YM} Cov(X, M^\ast) \nonumber \\
&= \theta_{YM} Cov(X, \theta_{MX} X + U_M) = \theta_{YM} \theta_{MX} Var(X) = \theta_{YM} \theta_{MX}~, \label{eq:causal.indirect.effect}
\end{align}

Finally, suppose that the goal is to build a ML model whose predictive performance is only informed by the spurious associations generated by the confounder. We can simulate data according to the twin network in Figure \ref{fig:twins}, so that,
\begin{align}
Cov(X, Y^\ast) &= Cov(X, \theta_{YC} C + U_Y) = \theta_{YC} Cov(X, C) \nonumber \\
&= \theta_{YC} Cov(\theta_{XC} C + U_X, C) = \theta_{YC} \theta_{XC} Var(C) = \theta_{YC} \theta_{XC}~. \label{eq:causal.confounding.effect}
\end{align}
\begin{figure}[!h]
{\scriptsize
$$
\xymatrix@-1.4pc{
(a) & & U_X \ar[dll] \ar[drr] & && (b) & & U_X \ar[dll] \ar[drr] & & & (c) & & U_X \ar[dll] \ar[drr] & &&\\
*+[F-:<10pt>]{X} \ar[ddd] \ar[ddr] & & U_C \ar[dl] \ar[dr] & & *+[F-:<10pt>]{X} \ar[ddd] \ar[ddl] & *+[F-:<10pt>]{X} \ar[ddd] \ar[ddr] & & U_C \ar[dl] \ar[dr] & & *+[F-:<10pt>]{X} \ar[ddl] & *+[F-:<10pt>]{X} \ar[ddd] \ar[ddr] & & U_C \ar[dl] \ar[dr] & & *+[F-:<10pt>]{X}  \\
& *+[F-:<10pt>]{C} \ar[ul] \ar[ddl] \ar[d] & U_M \ar[dl] \ar[dr] & *+[F-:<10pt>]{C} \ar[d] \ar[ur] && & *+[F-:<10pt>]{C} \ar[ul] \ar[ddl] \ar[d] & U_M \ar[dl] \ar[dr] & *+[F-:<10pt>]{C} \ar[ru] & & & *+[F-:<10pt>]{C} \ar[ul] \ar[ddl] \ar[d] & U_M \ar[dl] \ar[dr] & *+[F-:<10pt>]{C} \ar[rdd] \ar[d] \ar[ur] &&  \\
& *+[F-:<10pt>]{M} \ar[dl] && *+[F-:<10pt>]{M} &&& *+[F-:<10pt>]{M} \ar[dl] && *+[F-:<10pt>]{M^\ast} \ar[dr] & & & *+[F-:<10pt>]{M} \ar[dl] && *+[F-:<10pt>]{M} && \\
*+[F-:<10pt>]{Y} & & U_Y \ar[ll] \ar[rr] & & *+[F-:<10pt>]{Y^\ast} & *+[F-:<10pt>]{Y} & & U_Y \ar[ll] \ar[rr] & & *+[F-:<10pt>]{Y^\ast} & *+[F-:<10pt>]{Y} & & U_Y \ar[ll] \ar[rr] & & *+[F-:<10pt>]{Y^\ast} \\
}
$$}
\vskip -0.1in
  \caption{Twin network approach for the causal prediction tasks.}
  \label{fig:twins}
\end{figure}
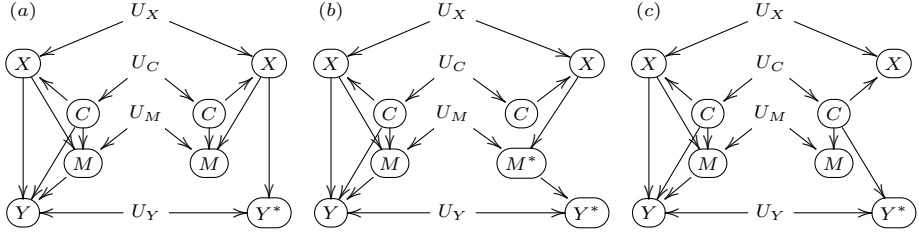

Similarly to the anticausal prediction task case, alternative interventions based on SWIGs can also be used. Figure \ref{fig:causal.node.splitting} shows the respective SWIGs for the generation of counterfactual responses.
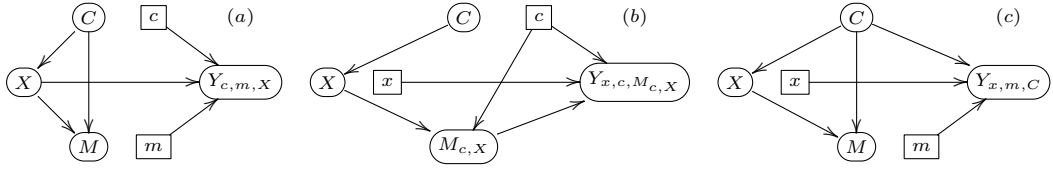
\begin{figure}[!h]
{\scriptsize
$$
\xymatrix@-1.1pc{
& *+[F-:<10pt>]{C} \ar[dd] \ar[dl] & *+[F]{c} \ar[dr] & (a) & & & *+[F-:<10pt>]{C} \ar[dll] & *+[F]{c} \ar[dr] \ar[ddl] & (b) & & & *+[F-:<10pt>]{C} \ar[dll] \ar[drr] \ar[dd] & & (c) \\
*+[F-:<10pt>]{X} \ar[dr] \ar[rrr] & & & *+[F-:<10pt>]{Y_{c,m,X}} & *+[F-:<10pt>]{X} \ar[drr] & *+[F]{x} \ar[rrr] & & & *+[F-:<10pt>]{Y_{x,c,{M_{c,X}}}} & *+[F-:<10pt>]{X} \ar[drr] & *+[F]{x} \ar[rrr] & & & *+[F-:<10pt>]{Y_{x,m,C}} \\
& *+[F-:<10pt>]{M} & *+[F]{m} \ar[ur] & & & & *+[F-:<10pt>]{M_{c,X}} \ar[urr] & & & & & *+[F-:<10pt>]{M} & *+[F]{m} \ar[ur] & & \\
}
$$}
\vskip -0.1in
  \caption{SWIGs for the causal predictive tasks.}
  \label{fig:causal.node.splitting}
\end{figure}

Direct calculation of the covariances shows that, $Cov(X, Y_{c,m,X}) = \theta_{YX}$ for the SWIG in panel a, $Cov(X, Y_{x,c,{M_{c,X}}}) = \theta_{YM} \, \theta_{MX}$ for the SWIG in panel b, and $Cov(X, Y_{x,m,C}) = \theta_{XC} \, \theta_{YC}$ for the SWIG in panel c.

Observe, that alternative interventions where we intervene on the features will not recover the correct associations. To illustrate this point, consider the simplified situation where we are interested in the direct causal effect, $\theta_{YX}$, in a model containing a confounder but no mediator. For the interventions presented in Figure \ref{fig:confoundingtwin2}a we have that,
\begin{align*}
Cov(X^\ast, Y^\ast) &= Cov(X^\ast, \, \theta_{YX} X^\ast + \theta_{YC} C + U_Y) \\
&= \theta_{YX} \, Var(X^\ast) + \theta_{YC} Cov(X^\ast, C) \\
&= \theta_{YX} \, Var(U_X) + \theta_{YC} Cov(U_X, C) \\
&= \theta_{YX} \, Var(U_X) \\
&= \theta_{YX} \, (1 - \theta_{XC}^2)~, \\
\end{align*}
where the last equality follows from the fact that $Var(U_X) = (1 - \theta_{XC}^2)$ since $1 = Var(X) = Var(\theta_{XC} C + U_X) = \theta_{XC}^2 Var(C) + Var(U_X) = \theta_{XC}^2 + Var(U_X)$. Similarly, even for the intervention in Figure \ref{fig:confoundingtwin2}b we still have that,
\begin{align*}
Cov(X^\ast, Y^\ast) &= Cov(X^\ast, \, \theta_{YX} X^\ast + U_Y) \\
&= \theta_{YX} \, Var(X^\ast) = \theta_{YX} \, Var(U_X) = \theta_{YX} \, (1 - \theta_{XC}^2)~. \\
\end{align*}

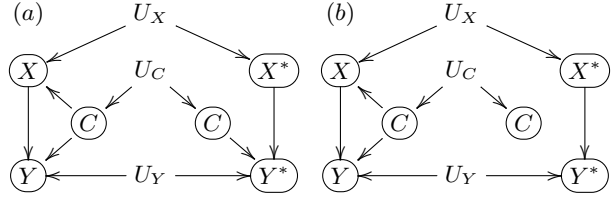
\begin{figure}[!h]
{\footnotesize
$$
\xymatrix@-1.4pc{
(a) & & U_X \ar[dll] \ar[drr] & && (b) & & U_X \ar[dll] \ar[drr] & &\\
*+[F-:<10pt>]{X} \ar[dd] & & U_C \ar[dl] \ar[dr] & & *+[F-:<10pt>]{X^\ast} \ar[dd] & *+[F-:<10pt>]{X} \ar[dd] & & U_C \ar[dl] \ar[dr] & & *+[F-:<10pt>]{X^\ast} \ar[dd] \\
& *+[F-:<10pt>]{C} \ar[ul] \ar[dl] &  & *+[F-:<10pt>]{C} \ar[rd] && & *+[F-:<10pt>]{C} \ar[ul] \ar[dl] & & *+[F-:<10pt>]{C} & \\
*+[F-:<10pt>]{Y} & & U_Y \ar[ll] \ar[rr] & & *+[F-:<10pt>]{Y^\ast} & *+[F-:<10pt>]{Y} & & U_Y \ar[ll] \ar[rr] & & *+[F-:<10pt>]{Y^\ast} \\
}
$$}
  \caption{Alternative model modifications for the confounding only examples.}
  \label{fig:confoundingtwin2}
\end{figure}

These examples illustrate that for causal prediction tasks, only interventions that do not modify $X$ can generate associations that recover the causal effects of interest.

\noindent \textbf{Remarks:} The fact that the causality-aware approach requires the computation of counterfactual responses, $Y^\ast$, implies that, contrary to anticausal prediction tasks (which requires the computation of counterfactual features, $X^\ast$, and where it is possible to estimate counterfactual features for both the training and test sets without having access to the test set responses), causal prediction tasks require access to the test set responses, $Y_{ts}$, in order to estimate the causal effects and residuals needed for the computation of the counterfactual test set responses, $Y^\ast_{ts}$. Since, in practice, $Y_{ts}$ is unavailable (as it is the quantity we want to predict) it follows that the approach cannot be used to generate, for example, stable predictions w.r.t. unknown shifts in target populations, as was done in the anticausal tasks. In causal prediction tasks, and under the assumption of no dataset shifts between the training and target populations, the causality-aware approach can still be used to estimate the predictive performance that is due to (or is free from) the influence of sensitive variables. For instance, we still can split our development data into independent and identically distributed training and validation sets and then compute counterfactual versions of the training and validation responses, in order to generate causality-aware predictions that can still be used to answer important questions such as, for example: ``what would the predictive performance of the learner be, had the (in)direct path not contributed to the association between the features and the response?" or ``what would the predictive performance of the learner be, had the observed confounders not biased the data?"

\subsection{Causality-aware predictions in causal prediction tasks - the multivariate case}

\begin{theorem}
Consider a causal prediction task:

(i) Suppose the interest focus on the causal effects generated by the paths in the path set $\bfX \rightarrow Y$. If $Y^\ast$ is given by $Y^\ast = \bfGamma_{YX} \, \bfX + W_Y$, then $Cov(Y^\ast, \bfX) = \bfGamma_{YX} \, Cov(\bfX)$.

(ii) Suppose the interest focus on the causal effects generated by the paths in the path set $\bfX \rightarrow \bfM \rightarrow Y$. If  $Y^\ast$ is given by $Y^\ast = \bfGamma_{YM} \, \bfM^\ast + W_Y$, and $\bfM^\ast = \bfGamma_{MX} \bfX + \bfW_M$, then $Cov(Y^\ast, \bfX) = \bfGamma_{YM} \, \bfGamma_{MX} \, Cov(\bfX)$.

(iii) Suppose the interest focus on the spurious associations generated by the paths in the path set $\bfX \leftarrow \bfC \rightarrow Y$. If  $Y^\ast$ is given by $Y^\ast = \bfGamma_{YC} \, \bfC + W_Y$, then $Cov(Y^\ast, \bfX) = \bfGamma_{YC} \, Cov(\bfC) \, \bfGamma_{XC}^T$.
\end{theorem}

\begin{proof}
$ $

Result \textit{i}: If $Y^\ast = \bfGamma_{YX} \, \bfX + W_Y$, then,
\begin{align*}
Cov(Y^\ast, \bfX) &= Cov(\bfGamma_{YX} \, \bfX + W_Y, \bfX) \\
&= \bfGamma_{YX} \, Cov(\bfX, \bfX) \\
&= \bfGamma_{YX} \, Cov(\bfX)
\end{align*}

Result \textit{ii}: If $Y^\ast = \bfGamma_{YM} \, \bfM^\ast + W_Y$ and $\bfM^\ast = \bfGamma_{MX} \bfX + \bfW_M$, then,
\begin{align*}
Cov(Y^\ast, \bfX) &= Cov(\bfGamma_{YM} \, \bfM^\ast + W_Y, \bfX) \\
&= \bfGamma_{YM} \, Cov(\bfX, \bfM^\ast)  \\
&= \bfGamma_{YM} \, Cov(\bfGamma_{MX} \bfX + \bfW_M, \bfX) \\
&= \bfGamma_{YM} \, \bfGamma_{MX} \, Cov(\bfX, \bfX) \\
&= \bfGamma_{YM} \, \bfGamma_{MX} \, Cov(\bfX)
\end{align*}

Result \textit{iii}: If $Y^\ast = \bfGamma_{YC} \, \bfC + W_Y$, then,
\begin{align*}
Cov(Y^\ast, \bfX) &= Cov(\bfGamma_{YC} \, \bfC + W_Y, \bfX) \\
&= \bfGamma_{YC} \, Cov(\bfC, \bfX)  \\
&= \bfGamma_{YC} \, Cov(\bfC, \bfGamma_{XC} \, \bfC + \bfW_X) \\
&= \bfGamma_{YC} \, Cov(\bfC, \bfC) \, \bfGamma_{XC}^T \\
&= \bfGamma_{YC} \, Cov(\bfC) \, \bfGamma_{XC}^T
\end{align*}
\end{proof}

Note that, in the univariate case, results (i), (ii), and (iii) in Theorem 2 reduce to the univariate results presented in equations (\ref{eq:causal.direct.effect}), (\ref{eq:causal.indirect.effect}), and (\ref{eq:causal.confounding.effect}), respectively (note that $Cov(\bfX)$ reduces to 1). Observe, as well, that results (\textit{i}) and (\textit{ii}) in Theorem 2 show that, in addition to the direct causal effect ($\bfGamma_{YX}$, in result \textit{i}) and the indirect causal effect ($\bfGamma_{YM} \, \bfGamma_{MX}$, in result \textit{ii}) the marginal covariances between the elements of $\bfX$ and $Y^\ast$ also depend on $Cov(\bfX)$. This makes sense, since $Cov(\bfX)$ captures the associations between the elements of $\bfX$. Note that for each element $X_j$ of $\bfX$, the operation $\bfGamma_{YX} \, Cov(\bfX)$ captures not only the association generated by the direct causal path $X_j \rightarrow Y^\ast$, but also the association generated by indirect and backdoor paths that start at $X_j$ and end at $Y^\ast$, but where the last node prior to $Y^\ast$ is another element $X_k$ of $\bfX$.




As an illustration, consider the DAG describing the causal prediction task in Figure \ref{fig:causal.example}a, where $Cov(\bfX)$,
\begin{equation*}
\begin{pmatrix}
1 & \theta_{{X_2}{X_1}} + \theta_{{X_1}{C_1}} \, \theta_{{X_2}{C_1}} \\
\theta_{{X_2}{X_1}} + \theta_{{X_1}{C_1}} \, \theta_{{X_2}{C_1}} & 1 \\
\end{pmatrix}~.
\end{equation*}
In this example, the association between $X_1$ and $X_2$,
\begin{equation*}
Cov(X_1, X_2) = \underbrace{\theta_{{X_2}{X_1}}}_{X_1 \rightarrow X_2} + \underbrace{\theta_{{X_1}{C_1}} \, \theta_{{X_2}{C}}}_{X_1 \leftarrow C_1 \rightarrow X_2}~,
\end{equation*}
is generated by the paths $X_1 \rightarrow X_2$ and $X_1 \leftarrow C_1 \rightarrow X_2$.
\begin{figure}[!h]
{\scriptsize
$$
\xymatrix@-1.5pc{
(a) & *+[F-:<10pt>]{C_1} \ar[dl] \ar[dddl] \ar[r] & *+[F-:<10pt>]{C_2} \ar[dddd] \ar[ddr] && (b) & *+[F-:<10pt>]{C_1} \ar[dl] \ar[dddl] \ar[r] & *+[F-:<10pt>]{C_2} \ar[dddd] & \\
*+[F-:<10pt>]{X_1} \ar[dd] \ar[drrr] & & & & *+[F-:<10pt>]{X_1} \ar[dd] \ar[drrr] \\
&&& *+[F-:<10pt>]{Y} &  &&& *+[F-:<10pt>]{Y^\ast} \\
*+[F-:<10pt>]{X_2} \ar[dr] \ar[urrr] & & & & *+[F-:<10pt>]{X_2} \ar[dr] \ar[urrr] \\
& *+[F-:<10pt>]{M_1} \ar[r] & *+[F-:<10pt>]{M_2} \ar[uur] & & & *+[F-:<10pt>]{M_1} \ar[r] & *+[F-:<10pt>]{M_2} \\
}
$$}
\vskip -0.1in
  \caption{A causal prediction task illustrative example.}
  \label{fig:causal.example}
\end{figure}
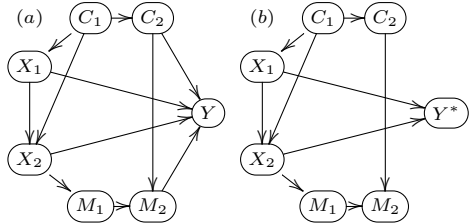
From result \textit{i} in Theorem 2 we have that,
\begin{align*}
Cov&(Y^\ast, \bfX) = \bfGamma_{YX} \, Cov(\bfX) = (\theta_{Y{X_1}}, \theta_{Y{X_2}}) \, Cov(\bfX) \\
&=
\begin{pmatrix}
\theta_{{Y}{X_1}} + \theta_{{X_2}{X_1}} \, \theta_{{Y}{X_2}} + \theta_{{X_1}{C_1}} \, \theta_{{X_2}{C_1}} \, \theta_{{Y}{X_2}} \\
\theta_{{Y}{X_2}} + \theta_{{X_2}{X_1}} \, \theta_{{Y}{X_1}} + \theta_{{X_2}{C_1}} \, \theta_{{X_1}{C_1}} \, \theta_{{Y}{X_1}}\\
\end{pmatrix}^T~, \\
&=
\begin{pmatrix}
Cov(Y^\ast, X_1) \\
Cov(Y^\ast, X_2) \\
\end{pmatrix}^T~.
\end{align*}
Note that the direct application of Wright's path analysis to the diagram in Figure \ref{fig:causal.example}b shows that we can decompose the covariance of $X_1$ and $Y^\ast$,
\begin{align*}
Cov(Y^\ast, X_1) = \underbrace{\theta_{{Y}{X_1}}}_{X_1 \rightarrow Y^\ast} + \underbrace{\theta_{{X_2}{X_1}} \, \theta_{{Y}{X_2}}}_{X_1 \rightarrow X_2 \rightarrow Y^\ast} + \underbrace{\theta_{{X_1}{C_1}} \, \theta_{{X_2}{C_1}} \, \theta_{{Y}{X_2}}}_{X_1 \leftarrow C_1 \rightarrow X_2 \rightarrow Y^\ast}~,
\end{align*}
in terms of the direct path $X_1 \rightarrow Y^\ast$, the indirect path $X_1 \rightarrow X_2 \rightarrow Y^\ast$, and the backdoor path $X_1 \leftarrow C_1 \rightarrow X_2 \rightarrow Y^\ast$. Similarly, the covariance of $X_2$ and $Y^\ast$,
\begin{align*}
Cov(Y^\ast, X_2) = \underbrace{\theta_{{Y}{X_2}}}_{X_2 \rightarrow Y^\ast} + \underbrace{\theta_{{X_2}{X_1}} \, \theta_{{Y}{X_1}}}_{X_2 \leftarrow X_1 \rightarrow Y^\ast} + \underbrace{\theta_{{X_2}{C_1}} \, \theta_{{X_1}{C_1}} \, \theta_{{Y}{X_1}}}_{X_2 \leftarrow C_1 \rightarrow X_1 \rightarrow Y^\ast}~,
\end{align*}
can be decomposed in terms of the direct path $X_2 \rightarrow Y^\ast$, and the backdoor paths $X_2 \leftarrow X_1 \rightarrow Y^\ast$ and $X_2 \leftarrow C_1 \rightarrow X_1 \rightarrow Y^\ast$. (Note that all the indirect and backdoor paths in this example either start at $X_1$ and end at $X_2$ before connecting to $Y^\ast$, or start at $X_2$ and end at $X_1$ before connecting to $Y^\ast$.)

\end{document}